\documentclass{article}

\usepackage[preprint]{neurips_2026}

\usepackage[utf8]{inputenc} 
\usepackage[T1]{fontenc}    
\usepackage{hyperref}       
\usepackage{url}            
\usepackage{booktabs}       
\usepackage{amsfonts}       
\usepackage{nicefrac}       
\usepackage{microtype}      
\usepackage{xcolor}         
\usepackage{amsmath,amsfonts,amssymb,dsfont,bm,amsthm}

\usepackage{subcaption}

\newtheorem{lemma}{Lemma}

\title{Missing data and cluster graphs: cluster-level missingness vs variable-level missingness}

\author{%
  Willow J. Scott \\
  Sorbonne Université, INSERM, \\
  Institut Pierre Louis d’Epidémiologie\\ 
  et de Santé Publique,\\ 
  F75012, Paris, France\\
  \texttt{willow.scott@inserm.fr} \\
  \And
  Eugenio Valdano \\
  Sorbonne Université, INSERM, \\
  Institut Pierre Louis d’Epidémiologie\\ 
  et de Santé Publique,\\ 
  F75012, Paris, France\\
  \texttt{eugenio.valdano@inserm.fr} \\
  \AND
  Charles K. Assaad \\
  Sorbonne Université, INSERM, \\
  Institut Pierre Louis d’Epidémiologie\\ 
  et de Santé Publique,\\ 
  F75012, Paris, France\\
  \texttt{charles.assaad@inserm.fr} \\
}

\usepackage{graphicx} 
\graphicspath{{./images/}}

\usepackage{ciphod}

\usepackage{mathbbol}

\usepackage{amsmath}
\usepackage{amsthm}

\usepackage{tikz}

\newtheorem{definition}{Definition}

\newtheorem{theorem}{Theorem}
\newtheorem{proposition}{Proposition}
\newtheorem{remark}{Remark}
\newtheorem{example}{Example}

\begin{document}

\maketitle

\begin{abstract}
Missing data is pervasive in many scientific domains such as public health, environmental science, and the social sciences. Recoverability from missing data is typically studied using fully specified variable-level missingness models despite that, in many applications, only coarse structural information is available, for instance when variables are grouped into clusters due to limited knowledge or interpretability reasons. In this paper, we investigate recoverability from such abstract representations. We introduce two classes of cluster-based missingness graphs: the m-C-DMG, which retains variable-specific missingness indicators, and the cm-C-DMG, which aggregates missingness mechanisms at the cluster level. We formalize the notion of compatibility between these abstract graphs and underlying variable-level missingness models, and study how this abstraction affects the recoverability of probabilistic and causal queries. In particular, we give graphical conditions of recovering the joint distribution as well as graphical conditions of recovering a macro causal effect.
\end{abstract}

\section{Introduction}

Missing data are ubiquitous in many scientific domains, including public health, environmental science, and the social sciences. Understanding when meaningful quantities can be recovered from incomplete data is therefore a fundamental problem. Existing methods effectively approximate the true distribution of incomplete data when missingness occurs randomly or can be explained by observed values. However, such methods produce biased results when data is Missing Not At Random (MNAR), which is the reality for most datasets in fields like health and epidemiology ~\citep{Rubin_1976, Peek_2018}. Methods for imputation with MNAR missingness include autoencoders \citep{Pereira_2024}, which provide little theoretical guarantees, and pattern-mixture models~\citep{Little_1993}, which rely on semi-parametric or parametric assumptions. This work focuses on a non-parametric alternative for dealing with MNAR data. In such settings, graphical models~\citep{Pearl_2000} provide a powerful framework to characterize recoverability of joint, conditional, and causal quantities. However, classical graph-based missing data theory addresses missing data at the level of fully specified variable-level models, where the missingness mechanism of each variable is explicitly represented \citep{Mohan_2013, Mohan_2014, Mohan_2021}.

In many modern applications, however, the underlying data-generating mechanisms are only partially known or are intentionally abstracted~\citep{Peters_2013, Inoue_2022, Assaad_2023, Anand_2025}. Variables are thus often grouped into clusters, either for interpretability or because only aggregated structural information is available. For example, clusters are necessary to deal with uncertainty around the relationships between sexual practices and use of prevention in studies on predictors of risk of acquiring Human Immunodeficiency Virus (HIV). In a conceptual graph, sexual behavior is treated as a cluster of measures such as condom use and number of sexual partners. Table~\ref{table:HIV} exhibits missingness in a synthetic dataset of clusters representing Sexual Behavior, Mental Health, and HIV Status. As is the case in most real-world data, missingness may occur in any of the variables and for a variety of reasons. One must thus reason about recoverability from abstract representations of missingness mechanisms rather than from fully specified variable-level graphs.

In this paper, we introduce and study two such abstractions: the \emph{m-C-DMG} and the \emph{cm-C-DMG}. The m-C-DMG retains variable-specific missingness indicators while representing causal relations between variables at the cluster level. The cm-C-DMG is a coarser abstraction that merges missingness indicators within each cluster into a single vertex. While both representations encode missingness mechanisms at a higher level of abstraction, they differ in the amount of information they preserve, raising the question of how this impacts recoverability.

In this paper, we formalize the compatibility between abstract graphs and variable-level missingness models and study when recoverability can be inferred from these abstractions. We provide results for the recovery of both joint distributions and a specific type of causal effects, known as macro causal effects, depending on the level of abstraction. Overall, this work contributes to the theoretical understanding of missing data problems under structural abstraction.

The remainder of the paper is organized as follows. Section~\ref{sec:background} presents the background and preliminaries. Section~\ref{sec:missing_cluster} introduces and compares the two proposed abstractions. Section~\ref{sec:joint_distribution} provides conditions for recovering joint distributions from these abstractions. Section~\ref{sec:causal_effect} gives conditions for recovering macro causal effects. Section~\ref{sec:sim} uses simulated data to exemplify the recovery results from the proposed abstractions. Finally, Section~\ref{sec:discussion} concludes the paper with a discussion.
All proofs are deferred to the Appendix.

\begin{table}[h]
\centering
\caption{Example of synthetic individual-level public-health data with missing values. Only the cluster-level graph between HIV status, mental health, and sexual behavior is assumed to be known, since the relationships within each cluster are unspecified.}
\begin{tabular}{c|cc|cc|c}
\hline
ID 
& \texttt{Anxiety scale} & \texttt{Depression scale} 
& \texttt{Condom use} & \texttt{\# of partners} 
& \texttt{HIV status} \\
\hline
1 & Moderate    & None         & NA               & 3   & NA \\
2 & None        & Minimal      & Never            & 1   & NA \\
3 & None        & Severe       & NA               & 5   & Negative \\
4 & Minimal     & None         & Almost always    & 21  & Acute \\
5 & Severe      & Minimal      & Sometimes        & NA  & Negative \\
6 & None        & None         & Always           & 1   & Chronic \\
\hline
\end{tabular}
\label{table:HIV}
\end{table}

\section{Background and preliminaries}
\label{sec:background}
Consider a structural causal model~\citep{Pearl_2000} that induces a causal graph, also known as an acyclic directed mixed graph (ADMG), and a distribution of interest $\Pr(\mathbb{V})$ over a set of variables $\mathbb{V}$. 
Let $\mathbb{V}$ be partitioned into $\mathbb{V}_o$ and $\mathbb{V}_m$, where
$\mathbb{V}_o \subseteq \mathbb{V}$ denotes the set of variables observed for all
records in the population, and $\mathbb{V}_m \subseteq \mathbb{V}$ denotes the set
of variables missing in at least one record. We refer to  $\mathbb{V}_m\cup \mathbb{V}_o$ as the set of substantive variables. A variable $V_i$ is said to be
\emph{fully observed} if $V_i \in \mathbb{V}_o$, and \emph{partially observed} if
$V_i \in \mathbb{V}_m$.
For each partially observed variable $V_i \in \mathbb{V}_m$, we associate two
additional variables: a missingness indicator $R_{V_i}$ and a proxy variable
$V_i^\ast$. The intended semantics is that $R_{V_i}$ indicates the missingness status of  $V_i^\ast$ in the dataset, and    is what we actually observe of $V_i$ in the available dataset, \ie, 
\begin{equation}
\label{eq:from_missing_to_not_missing}
V_i^\ast =
\begin{cases}
V_i & \text{if } R_{V_i} = 0,\\
NA & \text{if } R_{V_i} = 1.
\end{cases}
\end{equation}
In what follows we call $\Pr(\mathbb{V}^o, \mathbb{V}^*, R)$ the manifest distribution. It represents the distribution of what we observe in addition to the missingness indicator variables.

Missingness is typically categorized depending on the process that causes it ~\citep{Rubin_1976}:

\begin{itemize}
    \item Missing Completely at Random (MCAR): the probability of missingness in $V_i$ is completely independent of both the value of $V_i$ and all other variables in the dataset:
    $\indep{R_{V_i}}{(V_o \cup V_m)}{}$.
    For example, a folder from a file cabinet organized in alphabetical order is lost.
    
    \item Missing at Random (MAR): the probability of missingness may depend on observed variables, but not on the missing value itself:
    $\indepc{R_{V_i}}{V_m}{V_o}{}$.
    For example, a person may be less likely to disclose their condom use if they have a sexual partner living with HIV.
    
    \item Missing Not at Random (MNAR): the probability of missingness depends on unobserved data. For example, not using condoms may decrease the likelihood of being willing to share one's own test results for sexually transmitted infections, and condom usage is not observed.
\end{itemize}

Given a manifest distribution $\Pr(\mathbb{V}^*, \mathbb{V}^o, \mathbb{R})$, a query $Q$ is said to be \emph{recoverable} if we can compute a consistent estimate of $Q$ as if no data
were missing.
\ie, if it is decomposable into
terms that contain the missingness mechanism $R_{V_i} = 0$ of
every partially observed variable $V_i$ in $Q$~\citep{Mohan_2013}.

For MCAR data, all $V_i \in \mathbb{V}_m$, $\indep{R_{V_i}}{(\mathbb{V}_o \cup \mathbb{V}_m)}{}$. Therefore, the joint probability is recoverable trivially: $\Pr(\mathbb{V}) = \Pr(\mathbb{V} \mid R_{\mathbb{V}_m}=0) = \Pr(\mathbb{V}_o, {\mathbb{V}_m^\ast} \mid {R_{\mathbb{V}_m}} = 0)$. Since, conditional on ${R_{V_i}}=0$, $V_i^\ast$ is observed, considering only complete cases (listwise deletion) allows recovery of the joint probability $\Pr(V_i)$. For MAR data, all $V_i\in \mathbb{V}_m$, $\indepc{R_{V_i}}{\mathbb{V}_m}{\mathbb{V}_o}{}$. Thus, $\Pr(\mathbb{V}) = \Pr(\mathbb{V}_m \mid \mathbb{V}_o)\Pr(\mathbb{V}_o) = \Pr(\mathbb{V}_m \mid \mathbb{V}_o, R_{\mathbb{V}_m} = 0)\Pr(\mathbb{V}_o)=\Pr(\mathbb{V}_m^* \mid \mathbb{V}_o, R_{\mathbb{V}_m} = 0)\Pr(\mathbb{V}_o)$ ~\citep{Mohan_2013}, making it possible to estimate the distribution of $\Pr(\mathbb{V})$.

Recoverability is more complicated under MNAR, leading many to make the assumption that data is MAR or to use listwise deletion, leading to biased results ~\citep{Wothke_2000}. Going back to the synthetic data in Table~\ref{table:HIV}, suppose that mental health impacts willingness to report sexual behaviors. HIV status is also missing for some, due to an unobserved variable. Thus, MAR would allow recovery of information on sexual behavior, but such an assumption would bias estimation of the distribution of HIV status. 

There are existing methods of recovery under MNAR missingness mechanisms. m-ADMGs (otherwise called \textit{Missingness graphs} or \textit{m-graphs}) employ understanding of the causal structure behind missingness to recover queries ~\citep{Mohan_2013, Mohan_2014, Mohan_2021}.
An m-ADMG is a graph
$
\mathcal{G} = (\mathbb{V}, \mathbb{R}, \mathbb{V}^\ast, \mathbb{E}^\to, \mathbb{E}^{\longdashleftrightarrow}),
$
where $\mathbb{V}$ is a finite set of substantive variables. For each
$X \in \mathbb{V}$, $R_X \in \mathbb{R}$ denotes the missingness indicator
associated with $X$, with
$
\mathbb{R}=\{R_X:X\in\mathbb{V}\}.
$
Similarly, $X^\ast \in \mathbb{V}^\ast$ denotes the observed-data proxy of $X$,
with
$
\mathbb{V}^\ast=\{X^\ast:X\in\mathbb{V}\}.
$
The set $\mathbb{E}^\to$ contains directed edges, where an edge $X\to Y$
represents a direct causal effect of $X$ on $Y$. The set
$\mathbb{E}^{\longdashleftrightarrow}$ contains bidirected edges, where an edge
$X\longdashleftrightarrow Y$ represents the presence of an unobserved common cause of
$X$ and $Y$. Proxy vertices have no children and no bidirected edges. For each
$X\in\mathbb{V}$, the proxy variable $X^\ast$ is determined by $X$ and $R_X$:
$X^\ast=X$ when $R_X=0$, and $X^\ast=\mathrm{NA}$ when $R_X=1$.
Importantly, an m-ADMG provides a qualitative representation of the causal
relations in an underlying structural causal model (like an ADMG), as well as
missingness mechanisms and proxies of partially observed
variables. Like an ADMG, it is acyclic and may contain both directed and
bidirected edges.
It is also important to note that a single m-ADMG can be compatible with many manifest distributions.
Compatibility means that the manifest distribution satisfies the conditional independence and missingness constraints encoded by the m-ADMG. We can now redefine the notion of recoverability with respect to m-ADMGs.

\begin{definition}[Recoverability in m-ADMG, \cite{Mohan_2013}]
Given a $\mathcal{G}^m$, and a target query $Q$ defined on the variables in $\mathbb{V}$ , $Q$ is said to be recoverable in $\mathcal{G}^m$ if there exists an algorithm that produces a consistent estimate
of $Q$ for every manifest distribution  $\Pr(\mathbb{V}^*, \mathbb{V}^o, \mathbb{R})$ that is (1) compatible with $\mathcal{G}^m$ and (2) strictly positive over complete cases, \ie, $\Pr(\mathbb{V}^*, \mathbb{V}^o, \mathbb{R}) > 0$.
\end{definition}

Fully specified causal graphs are frequently unobtainable in epidemiological settings, where the interplay of biological, social, and behavioral factors lead to densely populated graphs with uncertainty around the exact relationships between variables. Partially specified graphs allow researchers to employ graphical causal methods, despite ambiguity in parts of these systems. In these cases, experts can sometimes specify a cluster-directed mixed graph (C-DMG), also known as a cluster graph~\citep{Tikka_2023, Anand_2023, Ferreira_2025}, which is an abstract causal graph whose vertices represent clusters of variables rather than individual variables. A C-DMG may contain directed edges and bidirected edges between clusters. A directed edge $C_i \to C_j$ indicates that at least one variable in cluster $C_i$ has a direct causal effect on at least one variable in cluster $C_j$. A bidirected edge $C_i \longdashleftrightarrow C_j$ indicates the presence of latent confounding between variables in the two clusters. Unlike ADMGs and m-ADMGs, C-DMGs may contain directed cycles, since feedback can arise at the cluster level even when the underlying variable-level graph is acyclic.
For illustration, Figure~\ref{fig:two_admgs_one_cdmg} presents a C-DMG with two of its compatible ADMGs.
However, missingness indicators have not yet been incorporated into cluster graphs, limiting their applicability in incomplete data settings.

We end this section by introducing some graphical terminology. Two vertices are said to be \emph{neighbors} if they are adjacent, that is, if they are connected by an edge, regardless of its type or orientation. We use the standard notions of parents, children, ancestors, and descendants induced by directed paths. A non-endpoint vertex on a path is called a \emph{collider} on that path if the two adjacent edges on the path both have arrowheads pointing into it. For example, in the path $ A \to B \leftarrow C, $ the vertex $B$ is a collider. In contrast, $B$ is a non-collider in paths of the form $ A \to B \to C, A \leftarrow B \to C,  A \leftarrow B \leftarrow C. $

\begin{figure}
    \centering
\begin{subfigure}{0.35\textwidth}
\centering
\begin{tikzpicture}[
    every node/.style={draw, circle, align=center, minimum width=0.2cm, minimum height=0.2cm},
    font=\footnotesize, scale=0.8
]

\node (z1) at (-2,0) {$Z_1$};
\node (z2) at (-2,2) {$Z_2$};
\node (x1) at (0, 0) {$X_1$};
\node (x2) at (0, 2) {$X_2$};
\node (y1) at (2, 0) {$Y_1$};
\node (y2) at (2, 2) {$Y_2$};

\draw[->,>=latex] (z1) -- (z2);
\draw[->,>=latex] (z1) -- (x1);
\draw[->,>=latex] (z1) -- (x2);
\draw[->,>=latex] (x1) to (x2);
\draw[->,>=latex] (x1) to (z2);
\draw[->,>=latex] (x2) to (y1);
\draw[->,>=latex] (y1) to (y2);

\end{tikzpicture}
\caption{}
\end{subfigure}
\hfill 
\begin{subfigure}{0.35\textwidth}
\centering
\begin{tikzpicture}[
    every node/.style={draw, circle, align=center, minimum width=0.2cm, minimum height=0.2cm},
    font=\footnotesize, scale=0.8
]

\node (z1) at (-2,0) {$Z_1$};
\node (z2) at (-2,2) {$Z_2$};
\node (x1) at (0, 0) {$X_1$};
\node (x2) at (0, 2) {$X_2$};
\node (y1) at (2, 0) {$Y_1$};
\node (y2) at (2, 2) {$Y_2$};

\draw[->,>=latex] (z1) -- (z2);
\draw[->,>=latex] (z1) -- (x1);
\draw[->,>=latex] (z2) -- (x1);
\draw[->,>=latex] (x1) to (x2);
\draw[->,>=latex] (x1) to (y1);
\draw[->,>=latex] (x2) to (y2);
\draw[->,>=latex] (x2) to (z1);
\draw[->,>=latex] (y1) to (y2);

\end{tikzpicture}
\caption{}
\end{subfigure}
\hfill
\begin{subfigure}{0.25\textwidth}
\centering
\begin{tikzpicture}[
    every node/.style={draw, circle, align=center, minimum width=0.2cm, minimum height=0.2cm},
    font=\footnotesize, scale=0.8
]

\node (Z) at (-2,0) {$C_{\mathbb Z}$};
\node (X) at (0, 0) {$C_{\mathbb X}$};
\node (Y) at (2, 0) {$C_{\mathbb Y}$};


\draw[->,>=latex] (X) -- (Y);

 \begin{scope}[transform canvas={yshift=-.25em}]
     \draw[->,>=latex] (Z) -- (X);
     \end{scope}
 \begin{scope}[transform canvas={yshift=.25em}]
     \draw[<-,>=latex] (Z) -- (X);
 \end{scope}


\end{tikzpicture}
\caption{}
\end{subfigure}
    \caption{Two ADMGs in (a) and (b) compatible with the same C-DMG in (c).}
    \label{fig:two_admgs_one_cdmg}
\end{figure}

\section{Missingness in cluster graphs}
\label{sec:missing_cluster}

A central objective of this paper is to investigate how missingness mechanisms can be incorporated into cluster graphs. To this end, we introduce two abstract missingness representations, which differ in the level of detail they retain about missingness indicators:

\begin{itemize}
    \item \textbf{missingness-C-DMG (m-C-DMG)}: an abstract graph in which each variable inside a cluster is associated with its own missingness indicator;
    \item \textbf{clustered-missingness-C-DMG (cm-C-DMG)}: a coarser abstract graph in which a single missingness vertex represents the missingness mechanism of the whole cluster.
\end{itemize}

This distinction is crucial because the two abstractions do not encode the same amount of information. The \emph{m-C-DMG} preserves variable-specific missingness information as it relates to each cluster, whereas the \emph{cm-C-DMG} merges these distinctions into a single cluster-level missingness vertex. As a consequence, the \emph{cm-C-DMG} is generally less informative than the \emph{m-C-DMG}.

\begin{definition}[m-C-DMG]
Let $\mathbb{C}=\{C_1,\dots,C_K\}$ be a partition of a finite set of substantive
variables $\mathbb{V}$. An \emph{m-C-DMG} is a graph
$
\mathcal{G}^{\mathbb{C},m}
=
(\mathbb{C},\mathbb{R},\mathbb{V}^\ast,\mathbb{E}^\to,\mathbb{E}^{{\longdashleftrightarrow}}),
$
where $\mathbb{C}$ is the set of cluster vertices,
$\mathbb{R}=\{R_X:X\in\mathbb{V}\}$ is the set of variable-level missingness
indicators, and $\mathbb{V}^\ast=\{X^\ast:X\in\mathbb{V}\}$ is the set of proxy
variables. The edges are defined as follows:
\begin{itemize}
    \item The set $\mathbb{E}^\to$ contains directed edges between elements of
$\mathbb{C}\cup\mathbb{R}\cup\mathbb{V}^\ast$. A directed edge
$C_{\mathbb X}\to C_{\mathbb Y}$ between two cluster vertices means that there
exist variables $X_i\in C_{\mathbb X}$ and $Y_j\in C_{\mathbb Y}$ such that
$X_i$ has a direct causal effect on $Y_j$ at the variable level. Similarly, a
directed edge $C_{\mathbb X}\to R_Y$ means that there exists
$X_i\in C_{\mathbb X}$ such that $X_i$ has a direct causal effect on the
missingness indicator $R_Y$.
\item The set $\mathbb{E}^{{\longdashleftrightarrow}}$ contains bidirected edges between
elements of $\mathbb{C}\cup\mathbb{R}$. A bidirected edge
$C_{\mathbb X}\longdashleftrightarrow C_{\mathbb Y}$ means that there exist variables
$X_i\in C_{\mathbb X}$ and $Y_j\in C_{\mathbb Y}$ that share an unobserved
common cause at the variable level.
\end{itemize}

\end{definition}


Specifying an m-C-DMG can be difficult because it requires distinguishing the missingness mechanisms of individual variables and their relations to the rest of the graph. It can therefore be useful to cluster missingness indicators and work with a coarser representation. However, allowing an arbitrary partition of the missingness indicators would make it difficult to track which missingness mechanism corresponds to which substantive variables. For this reason, we focus on a structured form of clustering in which missingness indicators are merged only when their associated variables already belong to the same substantive cluster. We call the m-C-DMG where each missingness indicator $R_{C_{\mathbb{X}}}$ represents the missingness mechanism of an entire cluster $C_{\mathbb{X}}$ a cm-C-DMG.

\begin{definition}[cm-C-DMG]
Let $\mathbb{C}=\{C_1,\dots,C_K\}$ be a partition of a finite set of substantive
variables $\mathbb{V}$. A \emph{cm-C-DMG} is a graph
$
\mathcal{G}^{\mathbb{C},cm}
=
(\mathbb{C},\mathbb{R}_{\mathbb{C}},\mathbb{V}^\ast,
\mathbb{E}^\to,\mathbb{E}^{\longdashleftrightarrow}),
$
where $\mathbb{C}$ is the set of cluster vertices,
$\mathbb{R}_{\mathbb{C}}=\{R_C:C\in\mathbb{C}\}$ is the set of cluster-level
missingness indicators, and
$\mathbb{V}^\ast=\{X^\ast:X\in\mathbb{V}\}$ is the set of variable-level proxy
vertices.
For each cluster $C\in\mathbb{C}$, the node $R_C$ represents the missingness
indicators associated with variables in $C$, namely $\{R_X:X\in C\}$. Missingness
indicators associated with variables in different clusters are not merged. The edges are defined as follows:

\begin{itemize}
    \item The set $\mathbb{E}^\to$ contains directed edges between elements of
$\mathbb{C}\cup\mathbb{R}_{\mathbb{C}}\cup\mathbb{V}^\ast$. A directed edge
$C_{\mathbb X}\to C_{\mathbb Y}$ means that there exist variables
$X_i\in C_{\mathbb X}$ and $Y_j\in C_{\mathbb Y}$ such that $X_i$ has a direct
causal effect on $Y_j$ at the variable level. A directed edge
$C_{\mathbb X}\to R_{C_{\mathbb Y}}$ means that there exist
$X_i\in C_{\mathbb X}$ and $Y_j\in C_{\mathbb Y}$ such that $X_i$ has a direct
causal effect on $R_{Y_j}$ at the variable level.

\item The set $\mathbb{E}^{\longdashleftrightarrow}$ contains bidirected edges between
elements of $\mathbb{C}\cup\mathbb{R}_{\mathbb{C}}$. A bidirected edge
$C_{\mathbb X}\longdashleftrightarrow C_{\mathbb Y}$ means that there exist variables
$X_i\in C_{\mathbb X}$ and $Y_j\in C_{\mathbb Y}$ that share an unobserved
common cause at the variable level.  
\end{itemize}

\end{definition}


Unlike m-ADMGs and similarly to C-DMGs, m-C-DMGs and cm-C-DMGs may contain cycles.

In an m-CDMG or cm-C-DMG, for every $X \in \mathbb{V}$, the proxy vertex $X^\ast$ has exactly two parents:
    the cluster vertex $C_{\mathbb{X}}$ containing $X$, and the missingness indicator.
    It is also important to highlight that proxy vertices have no children and no bidirected edges.
    The only difference between the two definitions is that in the former, missingness indicator represents the missingness mechanism of the single variable $X$, whereas in the latter, a missingness indicator represents that of an entire cluster.

\begin{figure}
    \centering
\begin{subfigure}{0.45\textwidth}
\begin{tikzpicture}[
    every node/.style={draw, circle, align=center, minimum width=0.2cm, minimum height=0.2cm},
    font=\footnotesize, scale=0.8
]

\node (Z) at (-2,0) {$C_{\mathbb Z}$};
\node (X) at (0, 0) {$C_{\mathbb X}$};
\node (Y) at (2, 0) {$C_{\mathbb Y}$};
\node (Y1*) at (-1, -1) {$Y^*$};
\node (X1*) at (0, 1.5) {$X_1^*$};
\node (X2*) at (1.5, 1.5) {$X_2^*$};

\node (ry1) at (-4,0) {$R_Y$};
\node (rx1) at (-1.5,1.5) {$R_{X_1}$};
\node (rx2) at (3,1.5) {$R_{X_2}$};

\draw[->,>=latex] (X) -- (Y);

 \begin{scope}[transform canvas={yshift=-.25em}]
     \draw[->,>=latex] (Z) -- (X);
     \end{scope}
 \begin{scope}[transform canvas={yshift=.25em}]
     \draw[<-,>=latex] (Z) -- (X);
 \end{scope}	

\draw[->,>=latex] (Y) -- (Y1*);
\draw[->,>=latex] (ry1) -- (Y1*);

\draw[->,>=latex] (X) -- (X1*);
\draw[->,>=latex] (rx1) -- (X1*);
\draw[->,>=latex] (Z) -- (rx1);

\draw[->,>=latex] (X) -- (X2*);
\draw[->,>=latex] (rx2) -- (X2*);

\draw[<->,>=latex, dashed] (ry1) to [out=80,in=100, looseness=1] (Z);

\draw[->,>=latex] (X) to[out=-70,in=-20,looseness=1.8] (X);
\draw[->,>=latex] (Y) to[out=70,in=20,looseness=1.8] (Y);
\draw[->,>=latex] (Z) to[out=70,in=20,looseness=1.8] (Z);

\end{tikzpicture}
\caption{}
\end{subfigure}
\hfill 
\begin{subfigure}{0.45\textwidth}
\begin{tikzpicture}[
    every node/.style={draw, circle, align=center, minimum width=0.2cm, minimum height=0.2cm},
    font=\footnotesize, scale=0.8
]

\node (Z) at (-2,0) {$C_{\mathbb Z}$};
\node (X) at (0, 0) {$C_{\mathbb X}$};
\node (Y) at (2, 0) {$C_{\mathbb Y}$};
\node (Y*) at (-1, -1) {$C_{\mathbb Y}^*$};
\node (X*) at (0, 1.5) {$C_{\mathbb X}^*$};

\node (ry) at (-4,0) {$R_{C_{\mathbb{Y}}}$};
\node (rx) at (-1.5,1.5) {$R_{C_{\mathbb{X}}}$};

\draw[->,>=latex] (X) -- (Y);

 \begin{scope}[transform canvas={yshift=-.25em}]
     \draw[->,>=latex] (Z) -- (X);
     \end{scope}
 \begin{scope}[transform canvas={yshift=.25em}]
     \draw[<-,>=latex] (Z) -- (X);
 \end{scope}	

\draw[->,>=latex] (Y) -- (Y*);
\draw[->,>=latex] (ry) -- (Y*);

\draw[->,>=latex] (X) -- (X*);
\draw[->,>=latex] (rx) -- (X*);
\draw[->,>=latex] (Z) -- (rx);

\draw[<->,>=latex, dashed] (ry) to [out=80,in=100, looseness=1] (Z);

\draw[->,>=latex] (X) to[out=70,in=20,looseness=1.8] (X);
\draw[->,>=latex] (Y) to[out=70,in=20,looseness=1.8] (Y);
\draw[->,>=latex] (Z) to[out=70,in=20,looseness=1.8] (Z);
\end{tikzpicture}





\caption{}
\end{subfigure}
    \caption{An m-C-DMG (a) and a cm-C-DMG (b), representing conceptual graphs of our running example, with the following clusters: $C_\mathbb{X}$: Sexual Behaviors, $C_\mathbb{Z}$: Mental Health, $C_\mathbb{Y}$: HIV Status. This graph shows a cyclical relationship between mental health and sexual behavior~\citep{Alvy_2011, Mustanski_2007, Ocleirigh_2013} and the impact of sexual behavior on HIV status~\citep{WHO_HIV_2025}, with missingness on sexual behavior and on HIV status.}
    \label{fig:m_clusters}
\end{figure}



\subsection{Compatibility classes}

Let $\compatible{\mathcal{G}^{\mathbb{C}, *m}}$ denote the set of variable-level m-ADMGs compatible with an abstract graph $\mathcal{G}_{\mathbb{C}}$.
Given an m-C-DMG $\mathcal{G}^{\mathbb{C}, m}$, let $\pi(\mathcal{G}^{\mathbb{C}, m})$ denote the associated cm-C-DMG obtained by merging, within each cluster, all variable-level missingness indicators into a single cluster-level missingness vertex.
Under this natural abstraction map, every variable-level m-ADMG compatible with $\mathcal{G}^{\mathbb{C}, m}$ is also compatible with $\pi(\mathcal{G}^{\mathbb{C}, m})$. Intuitively, this is because the cm-C-DMG forgets distinctions that are preserved by the m-C-DMG. Therefore, the compatibility class of the m-C-DMG is contained in that of the associated cm-C-DMG:

\begin{proposition}[Inclusion of compatibility classes]
\label{proposition:compatibility_class}
Let $\mathcal{G}^{\mathbb{C}, m}$ be an m-C-DMG and let $\mathcal{G}^{cm}=\pi(\mathcal{G}^{\mathbb{C}, m})$ be the associated cm-C-DMG. Then
$\compatible{\mathcal{G}^{\mathbb{C}, m}} \subseteq \compatible{\mathcal{G}^{\mathbb{C}, cm}}.$
Moreover, the inclusion is strict whenever at least one cluster contains two variables with distinct missingness mechanisms.
\end{proposition}

The inclusion is strict because different variable-level missingness mechanisms may induce the same cluster-level missingness abstraction. For example, consider a cluster $C=\{A,B\}$ with variable-level missingness indicators $R_A$ and $R_B$. The m-C-DMG can distinguish between situations where $R_A$ depends on one set of variables and $R_B$ depends on another. By contrast, the cm-C-DMG only records that the cluster $C$ is subject to missingness, without retaining which variable-specific indicator is involved. Hence, several distinct m-ADMGs compatible with the same cm-C-DMG may fail to be compatible with the corresponding m-C-DMG.

\subsection{Recoverability semantics}

Let $Q$ be a target query. We say that $Q$ is \emph{recoverable} in an abstract graph if $Q$ can be recovered uniformly for every variable-level m-ADMG compatible with the  abstract graph, \ie,
$
Q \text{ is recoverable in } \mathcal{G}^{\mathbb{C}, m}
 \Longleftrightarrow 
Q \text{ is recoverable in every } \mathcal{G}\in \compatible{\mathcal{G}^{\mathbb{C}, m}}
$
and
$
Q \text{ is recoverable in } \mathcal{G}^{\mathbb{C}, cm}
 \Longleftrightarrow 
Q \text{ is recoverable in every } \mathcal{G}\in \compatible{\mathcal{G}^{\mathbb{C}, cm}}.
$




In general, a target quantity \(Q\) is less likely to be recoverable from a cm-C-DMG than from an m-C-DMG, since the former is a coarser abstraction. This is formalized in the following proposition.

\begin{proposition}
\label{proposition:compatibility_recover}
Let \(\mathcal{G}^{\mathbb{C}, m}\) be an m-C-DMG and let
\(\mathcal{G}^{\mathbb{C}, cm}=\pi(\mathcal{G}^{\mathbb{C}, m})\) be the
associated cm-C-DMG.
Then, for any target query \(Q\),
$Q \text{ recoverable in } \mathcal{G}^{\mathbb{C}, cm}
\Rightarrow
Q \text{ recoverable in } \mathcal{G}^{\mathbb{C}, m}$.
The converse does not hold in general.
\end{proposition}

The above proposition shows that there may exist a cm-C-DMG \(\mathcal{G}^{\mathbb{C}, cm}\) and a target query \(Q\) such that \(Q\) is recoverable in every \(\mathcal{G}^m \in \compatible{\mathcal{G}^{\mathbb{C}, cm}}\), while \(Q\) is not recoverable from \(\mathcal{G}^{\mathbb{C}, cm}\) alone. This follows from the fact that the cm-C-DMG suppresses variable-specific missingness distinctions that may be essential for deriving a recoverability formula. 

\begin{example}
    Let \(C_{\mathbb{X}}=\{X_1,X_2\}\), and suppose the target query is \(Q=\Pr(X_1)\).
In an m-C-DMG, assume that only the missingness of \(X_2\) may depend on the cluster:
$C_{\mathbb{X}} \rightarrow R_{X_2}$,
and
$C_{\mathbb{X}} \not\rightarrow R_{X_1}$.
Thus, the missingness of \(X_1\) is not affected by variables in $C_{\mathbb{X}}$  and \(\Pr(X_1)\)
may be recoverable from the observed cases of \(X_1\).
After merging \(R_{X_1}\) and \(R_{X_2}\) into the cluster-level missingness vertex \(R_{C_{\mathbb{X}}}\),
the associated cm-C-DMG contains
$C_{\mathbb{X}} \to R_{C_{\mathbb{X}}}$.
This no longer indicates whether the cluster affects the missingness of \(X_1\),
of \(X_2\), or both. Hence the cm-C-DMG is compatible with cases in which the
missingness of \(X_1\) depends on \(C_{\mathbb{X}}\), including on \(X_1\) itself, making \(\Pr(X_1)\)
non-recoverable in general.
Therefore, \(\Pr(X_1)\) may be recoverable from the m-C-DMG but not from the
associated cm-C-DMG.
\end{example}




In the following, we focus on two queries: the joint distribution and macro causal effects. For each, we present recoverability results for both m-C-DMGs and cm-C-DMGs. We denote by \(\mathcal{G}^{\mathbb{C},*m}\) a graph that may be either an m-C-DMG or a cm-C-DMG.
For simplicity, all illustrative examples are presented using cm-C-DMGs. This allows us to emphasize the main ideas while working with the coarser and more compact abstraction.

\section{Recovering joint distribution}
\label{sec:joint_distribution}

Recovering the joint distribution provides the basis for many downstream probabilistic queries, including marginal and conditional distributions. In the presence of missing data, however, the joint distribution is not directly available from the observed data and may fail to be recoverable. In the following, we provide necessary conditions for recoverability.

\begin{theorem}[Necessary conditions for joint recoverability]
\label{theorem:necessary_recovering_joint_dist}
    Let $\mathcal{G}^{\mathbb{C}, *m}$ be an m-C-DMG or a cm-C-DMG. The necessary  condition for recovering any joint  distribution $\Pr(\mathbb{C})$ compatible with $\mathcal{G}^{\mathbb{C}, *m}$  is the absence of any vertex $C_{\mathbb{X}}\in \mathbb{C}^m$ such that: 
    \begin{itemize}
        \item $C_{\mathbb{X}}$ and its $R$-vertex are neighbors, or
        \item $C_{\mathbb{X}}$ and its $R$-vertex are connected by a path in which all intermediate vertices are colliders and elements of $\mathbb{C}^m\cup \mathbb{C}^o$.
    \end{itemize}
\end{theorem}

In the following, we provide sufficient conditions for recovering the joint distribution over all clusters. Theorem~\ref{theorem:sufficient_recovering_joint_dist} uses two notions: missingness-indicator blocks and the Markov blanket $\mathbb R^{(a)}$, both of which are defined in the Appendix.

\begin{theorem}[Sufficient conditions for joint recoverability]
\label{theorem:sufficient_recovering_joint_dist}
Let
$
\mathcal G^{\mathbb C,cm}
=
\left(
\mathbb C^o
\cup
\mathbb C^m
\cup
\mathbb R_{\mathbb C},
\mathbb E^{\mathbb C}
\right)
$
be a cm-C-DMG.
Let $\mathbb R^{(1)},\ldots,\mathbb R^{(K)}$
be the missingness-indicator blocks, and for any $1 \le a \le K$ let
 $MB^o\left(\mathbb R^{(a)}\right),
MB^m\left(\mathbb R^{(a)}\right)$ denote the fully observed and partially observed parts of
the substantive Markov blanket of $\mathbb R^{(a)}$.
If  there are no edges (including self-loops) between missingness-indicator vertices and for every missingness-indicator block
$\mathbb R^{(a)}$,
$
\mathbb R^{(a)}
\cap
R_{
    MB^m
    \left(
    \mathbb R^{(a)}
    \right)
}
=
\varnothing
$, 
then
the joint distribution $\Pr(\mathbb c)$ is recoverable using
$\mathcal G^{\mathbb C,cm}$ and given by
\begin{equation}
    \label{eq:recovery_formula}
\Pr(\mathbb c)
=
\frac{
\Pr
\left(
\mathbb R_{\mathbb C}=0,
\mathbb c
\right)
}{
\displaystyle
\prod_{a=1}^{K}
\Pr\left(
\mathbb R^{(a)}=0
\mid
MB^o
\left(
\mathbb R^{(a)}
\right),
MB^m
\left(
\mathbb R^{(a)}
\right),
R_{
MB^m
\left(
\mathbb R^{(a)}
\right)
}=0
\right)
}.
\end{equation}
\end{theorem}


It is important to note that the conditions of Theorem~\ref{theorem:sufficient_recovering_joint_dist} imply those of Theorem~\ref{theorem:necessary_recovering_joint_dist}, whereas the converse does not hold. Further details are provided in the Appendix.
These results are closely related to the recoverability results established by
\cite{Mohan_2014} for m-ADMGs. 
They show that analogous graphical
conditions can be formulated directly at the cluster level. Moreover, because
m-C-DMGs and cm-C-DMGs may contain cycles, the results suggest that the
acyclicity assumption imposed in standard ADMGs is not essential for this type of recoverability result.

In the following, we illustrate the above theorems on our running example and show that the joint distribution is recoverable in this case.

\begin{example}
    Consider the cm-C-DMG in Figure~\ref{fig:m_clusters} (b). In this graph, $C_{\mathbb X}$ is not adjacent to $R_{C_{\mathbb X}}$, and
there is no path from $C_{\mathbb X}$ to $R_{C_{\mathbb X}}$ whose
intermediate vertices are all colliders belonging to
$\mathbb{C}^m\cup\mathbb{C}^o$. Similarly, $C_{\mathbb Y}$ is not adjacent to
$R_{C_{\mathbb Y}}$, and no such collider path connects them. No
partially observed cluster violates the conditions of the theorem,
thus the joint distribution $\Pr(c_{\mathbb{X}}, c_{\mathbb{Y}}, c_{\mathbb{Z}})$ is recoverable.
\end{example}

Interestingly, based on the above theorem, we can show that recoverability of the joint distribution in the cm-C-DMG is equivalent to recoverability of the joint distribution in the m-C-DMG.
\begin{proposition}
\label{proposition:compatibility_recover_joint_distribution}
    Let \(\mathcal{G}^{\mathbb{C}, m}\) be an m-C-DMG and let
\(\mathcal{G}^{\mathbb{C}, cm}=\pi(\mathcal{G}^{\mathbb{C},m})\) be the
associated cm-C-DMG.
Then,  Theorem~\ref{theorem:necessary_recovering_joint_dist} is satisfied in $\mathcal{G}^{\mathbb{C}, cm}  
\iff$
Theorem~\ref{theorem:necessary_recovering_joint_dist} is satisfied in $\mathcal{G}^{\mathbb{C}, m}$. However, while Theorem~\ref{theorem:sufficient_recovering_joint_dist} is satisfied in $\mathcal{G}^{\mathbb{C}, cm}  
\implies$
Theorem~\ref{theorem:sufficient_recovering_joint_dist} is satisfied in $\mathcal{G}^{\mathbb{C}, m}$, the converse is not true.
\end{proposition}

\section{Recovering macro causal effects}
\label{sec:causal_effect}

So far, we have focused on the recoverability of probabilistic, non-causal queries, namely, joint distributions. We now turn to the recovery of macro causal effects, which is a central objective in many applications. For instance, we may want to assess the effect of sexual behavior on the risk of acquiring HIV, even when data is partially observed. Such questions require going beyond observational distributions and reasoning about interventions at the cluster level.

Before presenting the main results of this section, we introduce several key tools that will be used to derive recoverability conditions for macro causal effects. 

We begin by introducing the notion of blocked paths. Let
$\mathcal{G}^{\mathbb{C},*m}$ denote either an m-C-DMG or a cm-C-DMG. A path
between two vertices is said to be blocked by a set $C_{\mathbb W}$ if at least
one of the following conditions holds: either the path contains a non-collider
that belongs to $C_{\mathbb W}$, or the path contains a collider such that
neither the collider nor any of its descendants belongs to $C_{\mathbb W}$.
Otherwise, the path is said to be active given $C_{\mathbb W}$.

Two sets of vertices $C_{\mathbb X}$ and $C_{\mathbb Y}$ are said to be
$d$-separated by $C_{\mathbb W}$ in $\mathcal{G}^{\mathbb{C},*m}$ if every path
between a vertex in $C_{\mathbb X}$ and a vertex in $C_{\mathbb Y}$ is blocked
by $C_{\mathbb W}$. We write this as
$
\dsepc{\mathbb{X}}{\mathbb{Y}}{\mathbb{W}}{\mathcal{G}^{\mathbb{C},*m}}$.
The notion of $d$-separation was originally introduced for ADMGs~\citep{Pearl_2000}. It was recently
shown to be sound for C-DMGs, which may contain cycles~\citep{Ferreira_2025}. Since m-C-DMGs and
cm-C-DMGs are particular forms of C-DMGs, this soundness result applies directly to both graph classes considered here.

To reason about interventions, we also use mutilated graphs. Given a graph
$\mathcal{G}^{\mathbb{C},*m}$ and a set of vertices $C_{\mathbb X}$, we denote by
$\mathcal{G}^{\mathbb{C},*m}_{\overline{C_{\mathbb X}}}$ the graph obtained by
removing all incoming directed edges into vertices in $C_{\mathbb X}$. Similarly,
$\mathcal{G}^{\mathbb{C},*m}_{\underline{C_{\mathbb X}}}$ denotes the graph
obtained by removing all outgoing directed edges from vertices in
$C_{\mathbb X}$. $\mathcal{G}^{\mathbb{C},*m}_{\overline{C_{\mathbb X}}\underline{C_{\mathbb Z}}}$
denotes the graph obtained by removing incoming edges into $C_{\mathbb X}$ and
outgoing edges from $C_{\mathbb Z}$.

Using d-separation and the notion of mutilated graphs, we can now present the do-calculus, a set of rules that can be used to identify a causal effect. The do-calculus was originally introduced for ADMGs and m-ADMGs and was later
shown to be sound for C-DMGs~\citep{Ferreira_2025}. In the following, we present the corresponding rules for m-C-DMGs and cm-C-DMGs.

\begin{figure}
\centering
\begin{tikzpicture}[
    every node/.style={draw, circle, align=center, minimum width=0.2cm, minimum height=0.2cm},
    font=\footnotesize, scale=0.8
]

\node (Z) at (-2,0) {$C_{\mathbb Z}$};
\node (X) at (0, 0) {$C_{\mathbb X}$};
\node (Y) at (2, 0) {$C_{\mathbb Y}$};
\node (Y*) at (-1, -1) {$C_{\mathbb Y}^*$};

\node (ry) at (-4,0) {$R_{C_{\mathbb{Y}}}$};

\draw[->,>=latex] (X) -- (Y);

     \draw[->,>=latex] (Z) -- (X);

\draw[->,>=latex] (Y) -- (Y*);
\draw[->,>=latex] (ry) -- (Y*);

\draw[<->,>=latex, dashed] (ry) to [out=80,in=100, looseness=1] (Z);
\draw[<->,>=latex, dashed] (Z) to [out=80,in=100, looseness=0.7] (Y);

\draw[->,>=latex] (X) to[out=70,in=20,looseness=1.8] (X);
\draw[->,>=latex] (Y) to[out=70,in=20,looseness=1.8] (Y);
\draw[->,>=latex] (Z) to[out=70,in=20,looseness=1.8] (Z);

\end{tikzpicture}
    \caption{cm-DMG where the macro causal effect of $C_{\mathbb{X}}$ on $C_{\mathbb{Y}}$ is recoverable but the joint distribution is not.}
\label{fig:macro_causal_effect_recoverable}
\end{figure}



\begin{definition}[Rules of the do-calculus in m-C-DMGs and cm-C-DMGs]
    The three following rules of the do-calculus are the following:
	\begin{equation*}
		\begin{aligned}
			\textbf{Rule 1:}& \probac{\mathbb{c}_{\mathbb{y}}}{\interv{\mathbb{c}_{\mathbb{z}}},\mathbb{c}_{\mathbb{x}},\mathbb{c}_{\mathbb{w}}} = \probac{\mathbb{c}_{\mathbb{y}}}{\interv{\mathbb{c}_{\mathbb{z}}},\mathbb{c}_{\mathbb{w}}}
			&\text{if } \dsepc{\mathbb{C}_{\mathbb{Y}}}{\mathbb{C}_{\mathbb{X}}}{\mathbb{C}_{\mathbb{Z}}, \mathbb{C}_{\mathbb{W}}}{\mathcal{G}^{\mathbb{C},*m}_{\overline{\mathbb{C}_{\mathbb{Z}}}}}\\
			\textbf{Rule 2:}& \probac{\mathbb{c}_{\mathbb{y}}}{\interv{\mathbb{c}_{\mathbb{z}}},\interv{\mathbb{c}_{\mathbb{x}}},\mathbb{c}_{\mathbb{w}}} = \probac{\mathbb{c}_{\mathbb{y}}}{\interv{\mathbb{c}_{\mathbb{z}}},\mathbb{c}_{\mathbb{x}},\mathbb{c}_{\mathbb{w}}}
			&\text{if } \dsepc{\mathbb{C}_{\mathbb{Y}}}{\mathbb{C}_{\mathbb{X}}}{\mathbb{C}_{\mathbb{Z}}, \mathbb{C}_{\mathbb{W}}}{\mathcal{G}^{\mathbb{C},*m}_{\overline{\mathbb{C}_{\mathbb{Z}}}\underline{\mathbb{C}_{\mathbb{X}}}}}\\
			\textbf{Rule 3:}& \probac{\mathbb{c}_{\mathbb{y}}}{\interv{\mathbb{c}_{\mathbb{z}}},\interv{\mathbb{c}_{\mathbb{x}}},\mathbb{c}_{\mathbb{w}}} = \probac{\mathbb{c}_{\mathbb{y}}}{\interv{\mathbb{c}_{\mathbb{z}}},\mathbb{c}_{\mathbb{w}}}
			&\text{if } \dsepc{\mathbb{C}_{\mathbb{Y}}}{\mathbb{C}_{\mathbb{X}}}{\mathbb{C}_{\mathbb{Z}}, \mathbb{C}_{\mathbb{W}}}{\mathcal{G}^{\mathbb{C},*m}_{\overline{\mathbb{C}_{\mathbb{Z}}}\overline{\mathbb{C}_{\mathbb{X}}(\mathbb{C}_{\mathbb{W}})}}}\\
		\end{aligned}
	\end{equation*}
	where $\mathbb{C}_{\mathbb{X}}(\mathbb{C}_{\mathbb{W}})$ is the set of vertices in $\mathbb{C}_{\mathbb{X}}$ that are non-ancestors of any vertex in $\mathbb{C}_{\mathbb{W}}$ in the mutilated graph $\mathcal{G}^{\mathbb{C},*m}_{\overline{{\mathbb{C}_{\mathbb{Z}}}}}$.
\end{definition}

Having introduced all the necessary tools, we are now in a position to show how to recover macro causal effects from m-C-DMGs and cm-C-DMGs. The results presented in this section provide conditions under which such effects can be identified from the observed data using the graphical structure. This constitutes the main result of this section.

\begin{theorem}[Recoverability of macro causal effects in m-C-DMGs and cm-C-DMGs]
	\label{theorem:soundness_do_calculus_missingness}
A causal query \(Q\) is recoverable from an m-C-DMG or cm-C-DMG if there exists
a sequence of applications of do-calculus, probability manipulations, and Equation~\ref{eq:from_missing_to_not_missing}
that transforms \(Q\) into an expression involving only observable quantities (namely proxy variables,
missingness indicators, and fully observed variables) and containing no do-operators.
\end{theorem}

In the following, we illustrate the theorem with two examples.

\begin{example}
    Consider the cm-C-DMG in Figure~\ref{fig:m_clusters} (b). Suppose we are interested in the
macro causal effect
$\probac{c_{\mathbb Y}}{\interv{c_{\mathbb X}}}$.
The first step is to transform the partially observed outcome into its observed
proxy. Since $R_{C_{\mathbb Y}}$ is d-separated from $C_{\mathbb Y}$ in the
mutilated graph where $C_{\mathbb X}$ is intervened on, Rule 1 of the
do-calculus gives
$
\probac{c_{\mathbb Y}}{\interv{c_{\mathbb X}}}
=
\probac{c_{\mathbb Y}}{\interv{c_{\mathbb X}},R_{C_{\mathbb Y}}=0}.
$
Then, by Equation~\ref{eq:from_missing_to_not_missing},
$
\probac{c_{\mathbb Y}}{\interv{c_{\mathbb X}},R_{C_{\mathbb Y}}=0}
=
\probac{c_{\mathbb Y}^\ast}{\interv{c_{\mathbb X}},R_{C_{\mathbb Y}}=0}.
$

The next step is to remove the do-operator. In this graph, there is no open
non-causal path from $C_{\mathbb X}$ to $C_{\mathbb Y}$ after conditioning
on the observed outcome missingness status. Hence, by Rule 2 of the do-calculus,
$
\probac{c_{\mathbb Y}^\ast}{\interv{c_{\mathbb X}},R_{C_{\mathbb Y}}=0}
=
\probac{c_{\mathbb Y}^\ast}{c_{\mathbb X},R_{C_{\mathbb Y}}=0}.
$

Finally, since $C_{\mathbb X}$ is partially observed, we need to replace it by its
proxy. Notice that in our graph we have $C_{\mathbb Y}^\ast$ d-separated from $R_{C_{\mathbb X}}$
given $C_{\mathbb X}$ and $R_{C_{\mathbb Y}}$. So using Rule 1 of the do-calculus we can add  $R_{C_{\mathbb X}}=0$ to the expression. Then
Using again
Equation~\ref{eq:from_missing_to_not_missing}, we obtain
$
\probac{c_{\mathbb Y}^\ast}{c_{\mathbb X},R_{C_{\mathbb Y}}=0}
=
\probac{c_{\mathbb Y}^\ast}
{c_{\mathbb X}^\ast,R_{C_{\mathbb X}}=0,R_{C_{\mathbb Y}}=0}.
$
\end{example}

In the previous example, the missingness indicator of $C_{\mathbb{Y}}$ can be $d$-separated from $C_{\mathbb{Y}}$. This might suggest that such a separation is always required for recoverability.

\begin{example}
Consider the cm-C-DMG in Figure~\ref{fig:macro_causal_effect_recoverable}.
Suppose that we are interested in the macro causal effect
$\probac{C_{\mathbb{Y}}}{\interv{C_{\mathbb{X}}}}$.
The first step is to replace each partially observed variable that is not
intervened on by its observed proxy. By Rule 1 of the do-calculus,
$
\probac{c_{\mathbb{Y}}}{\interv{c_{\mathbb{X}}}}
=
\probac{c_{\mathbb{Y}}}{\interv{c_{\mathbb{X}}},R_{C_{\mathbb{Y}}}=0}.
$
Then, by Equation~\ref{eq:from_missing_to_not_missing},
$
\probac{c_{\mathbb{Y}}}{\interv{c_{\mathbb{X}}},R_{C_{\mathbb{Y}}}=0}
=
\probac{c_{\mathbb{Y}}^\ast}{\interv{c_{\mathbb{X}}},R_{C_{\mathbb{Y}}}=0}.
$

The next step is to remove the do-operator. We first use the law of total
probability to introduce $C_{\mathbb{Z}}$:
$
\probac{c_{\mathbb{Y}}^\ast}{\interv{c_{\mathbb{X}}},R_{C_{\mathbb{Y}}}=0}
=
\sum_{c_{\mathbb{Z}}}
\probac{c_{\mathbb{Y}}^\ast}
{\interv{c_{\mathbb{X}}},R_{C_{\mathbb{Y}}}=0,c_{\mathbb{Z}}}
\probac{c_{\mathbb{Z}}}
{\interv{c_{\mathbb{X}}},R_{C_{\mathbb{Y}}}=0}.
$
By Rule 2 of the do-calculus,
$
\probac{c_{\mathbb{Y}}^\ast}{\interv{c_{\mathbb{X}}},R_{C_{\mathbb{Y}}}=0}
=
\sum_{c_{\mathbb{Z}}}
\probac{c_{\mathbb{Y}}^\ast}
{c_{\mathbb{X}},R_{C_{\mathbb{Y}}}=0,c_{\mathbb{Z}}}
\probac{c_{\mathbb{Z}}}
{\interv{c_{\mathbb{X}}},R_{C_{\mathbb{Y}}}=0}.
$
Finally, by Rule 3,
$
\probac{c_{\mathbb{Y}}^\ast}{\interv{c_{\mathbb{X}}},R_{C_{\mathbb{Y}}}=0}
=
\sum_{c_{\mathbb{Z}}}
\probac{c_{\mathbb{Y}}^\ast}
{c_{\mathbb{X}},R_{C_{\mathbb{Y}}}=0,c_{\mathbb{Z}}}
\probac{c_{\mathbb{Z}}}
{R_{C_{\mathbb{Y}}}=0}.
$
\end{example}

Interestingly, the previous example shows that a macro causal effect can be recoverable even when the joint distribution is not. This highlights that recoverability of causal effects does not necessarily require recoverability of the full joint distribution.


\section{Simulated data}
\label{sec:sim}
To exemplify recovery of the joint distribution from a cm-C-DMG, we simulated $100$ datasets in two different settings\footnote{Code simulation study can be accessed at \url{https://github.com/CIPHOD/pyCIPHOD/tree/main/reproducibility/neurips2026}}. The two settings, referred to as Setting 1 and Setting 2, are based on two different m-ADMGs compatible with the cm-C-DMG in Figure~\ref{fig:m_clusters}b.
All variables are binary and generated according to a structural causal model. Each is sampled from a Bernoulli distribution as functions of their parents, with independent exogenous noise terms: {$V_i \sim \mathrm{Bernoulli}(\operatorname{expit}\!\left(\alpha_i + \sum_{j\in \mathrm{Pa}(V_i)} \beta_{ij}V_j +  \gamma_i U_i\right)),$} where $V_i \in \{0,1\}$ and random coefficients $\in \{-10, 10\}$. We introduced missingness into the datasets through a similar process. See additional details on data simulation in the Appendix.

We estimated the full joint distributions for both settings under four hypotheses: 1) MCAR, using a complete case analysis, 2) MAR, using the Horvitz-Thompson estimator~\citep{Horvitz_1952,Robins_1995}, 3) MNAR, using the true m-ADMG underlying the data-generating mechanism with the recovery criterion of \cite{Mohan_2014}, and 4) MNAR, using a cm-C-DMG with Equation~\ref{eq:recovery_formula}. Figure \ref{fig:sim_res} compares the estimated joint distribution with the true joint distribution, using the normalized total variation distance (NTVD)~\citep{Gibbs_2002}: 
$\frac{1}{2} \sum_{\mathbb{v}}|\Pr(\mathbb{v}) - \hat\Pr(\mathbb{v})|$ where $\Pr$ refers to the true joint distribution and $\hat\Pr$ refers to the estimated joint distribution. Results are visualized with $\pm 1$ standard deviation error bars.

\begin{figure}[t]
    \centering

    \begin{subfigure}{0.45\textwidth}
        \centering
        \includegraphics[width=0.85\linewidth]{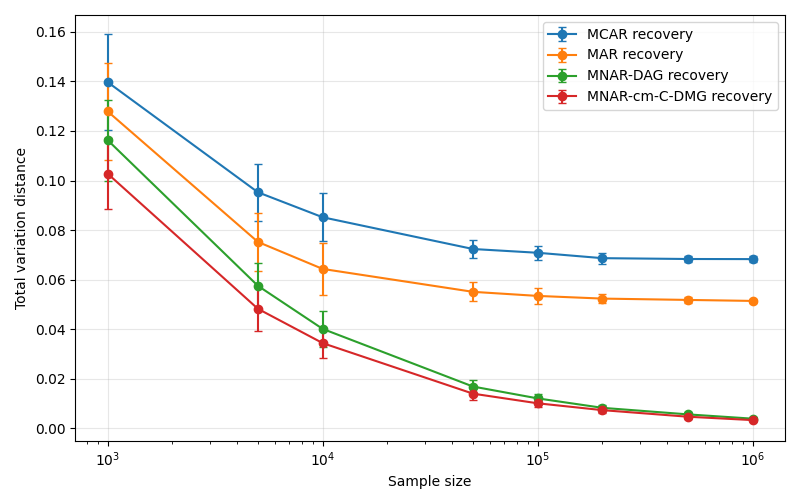}
        \caption{Joint distribution recovery in setting 1.}
        \label{fig:joint_setting1}
    \end{subfigure}
    \hfill
    \begin{subfigure}{0.45\textwidth}
        \centering
        \includegraphics[width=0.85\linewidth]{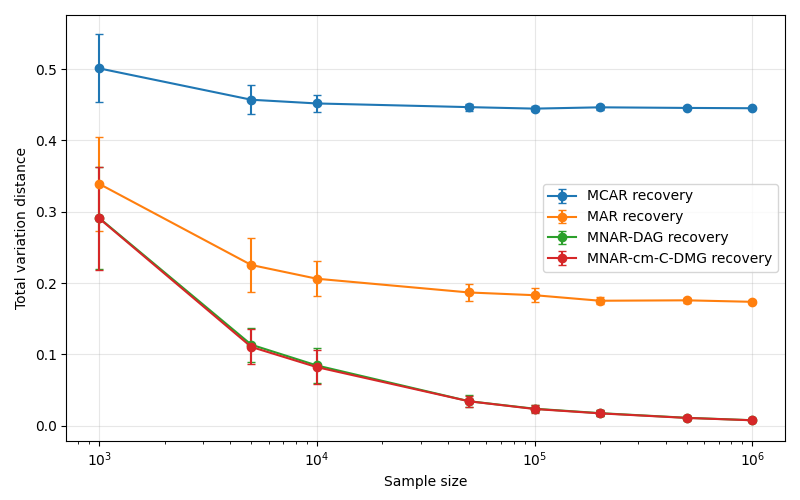}
        \caption{Joint distribution recovery setting 2.}
        \label{fig:joint_setting2}
    \end{subfigure}

    \begin{subfigure}{0.45\textwidth}
        \centering
        \includegraphics[width=0.85\linewidth]{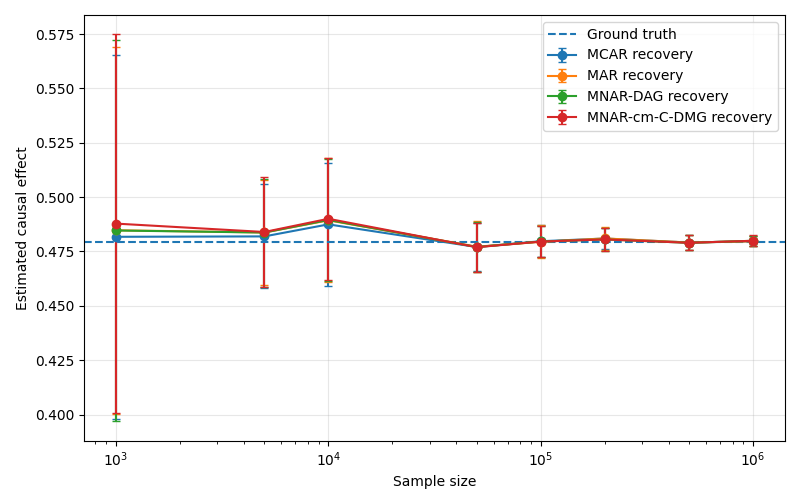}
        \caption{Macro causal effect estimation in setting 1.}
        \label{fig:effect1}
    \end{subfigure}
    \hfill
    \begin{subfigure}{0.45\textwidth}
        \centering
        \includegraphics[width=0.85\linewidth]{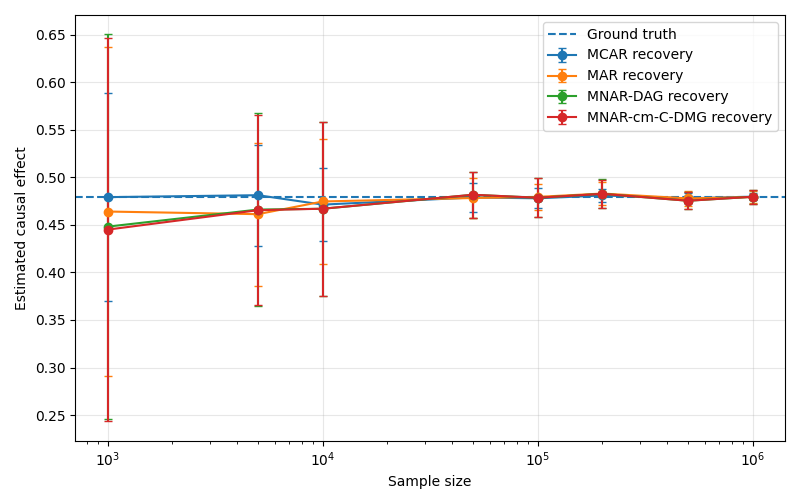}
        \caption{Macro causal effect estimation in setting 2.}
        \label{fig:effect2}
    \end{subfigure}
    
    \caption{Joint distribution recovery (a,b) and macro-level causal effect estimation (c,d) under two settings. Settings 1 and 2 correspond respectively to data-generating processes compatible with two distinct m-ADMGs, both compatible with the cm-C-DMG in Figure~\ref{fig:m_clusters}b.}
    \label{fig:sim_res}
\end{figure}

Figure~\ref{fig:sim_res} (a,b) shows that in Settings 1 and 2, the NTVD under the MNAR assumption converges to zero for both the fully specified m-ADMG and the cm-C-DMG as the sample size increases. Recovery based on the cm-C-DMG is as accurate as recovery based on the m-ADMG in Setting 2, and slightly more accurate in Setting 1, especially when the number of observations is relatively small. Further investigation is necessary to understand why this is the case. In contrast, under MAR and MCAR assumptions, the NTVD decreases initially but plateaus between approximately 0.05 and 0.45, even for large sample sizes.

Figure~\ref{fig:sim_res} (c,d) shows estimation of the macro causal effect of $C_{\mathbb{X}}$ on $C_{\mathbb{Y}}$ based on the estimated joint distributions. Perhaps counter-intuitively, estimation with MCAR recovery, which is less precise for joint distribution estimation, is the most accurate at small sample sizes in both Setting 1 and Setting 2. However, considering the results of Example 3, we note that the recovery of the macro causal effect is ultimately equivalent to assuming MCAR missingness for variables involved in causal effect estimation. As sample size grows, all recovery methods converge to the correct macro causal effect.

\section{Discussion}
\label{sec:discussion}

In this paper, we studied recoverability under abstract missingness representations. We introduced two cluster-based missingness graphs: the m-C-DMG, which preserves variable-specific missingness indicators, and the cm-C-DMG, which further aggregates missingness indicators associated with variables in the same cluster. We formalized compatibility between these abstract graphs and variable-level missingness models and showed how recoverability conclusions can be transferred across abstraction levels. In particular, recoverability from a cm-C-DMG implies recoverability from the corresponding m-C-DMG, while the converse does not hold in general. Moreover, we gave sufficient conditions and necessary conditions for recovering joint distributions and sufficient conditions for recovering macro causal effects.
Studying missingness through causal abstractions has several advantages. First, it provides a principled way to analyze settings where only coarse structural knowledge is available. Second, it clarifies how much information is lost when missingness mechanisms are aggregated, and when this loss matters for a specific query.

This work has several limitations. Because m-C-DMGs and cm-C-DMGs are abstractions, they may be conservative: a query that is not recoverable from the abstract graph may still be recoverable in some compatible variable-level graph. Moreover, cm-C-DMGs can hide distinctions between missingness indicators that are essential for recovery. Their usefulness thus depends on whether the chosen clustering preserves the missingness information relevant to the target query. Furthermore, we considered only a special case of cm-C-DMGs, where missingness indicators are clustered only when their associated variables belong to the same substantive cluster. In real applications, if experts cannot specify an m-C-DMG, the partitioning they can provide may not satisfy this simplifying requirement.

Several questions remain open. An important direction is to determine whether our recoverability results are complete. Another direction is to develop algorithms for recoverability directly on these abstract graphs. For non-causal queries, we focused mainly on recovering joint distributions; it would therefore be interesting to further investigate conditional distributions. Furthermore, it would be interesting to study the relative efficiency of recovery based on ADMGs compared with recovery based on cluster graphs in finite-sample settings. Finally, future work should extend the framework to the recovery of micro causal effects under causal abstractions, building on recent work on identifiability from abstract causal graphs~\citep{Assaad_2024}.

\section*{Societal impacts}
This paper presents work whose goal is to advance the field of Machine Learning and Epidemiology, and the limits of what can be inferred from data. There are many potential societal consequences of our work, none which we feel must be specifically highlighted here.

\section*{Acknowledgement}
We thank anonymous reviewers for their constructive comments. This work was supported by the CIPHOD project (ANR-23-CPJ1-0212-01) as well as by the PP-IMPACT project.

\bibliographystyle{plainnat}
\bibliography{references}

@book{Pearl_2000,
author = {Pearl, Judea},
title = {Causality: Models, Reasoning and Inference},
year = {2009},
isbn = {052189560X},
publisher = {Cambridge University Press},
address = {USA},
edition = {2nd}
}

@inproceedings{Mohan_2013,
 author = {Mohan, Karthika and Pearl, Judea and Tian, Jin},
 booktitle = {Advances in Neural Information Processing Systems},
 editor = {C.J. Burges and L. Bottou and M. Welling and Z. Ghahramani and K. Weinberger},
 pages = {},
 publisher = {Curran Associates, Inc.},
 title = {Graphical Models for Inference with Missing Data},
 volume = {26},
 year = {2013}
}

@article{Mohan_2021,
author = {Karthika Mohan and Judea Pearl},
title = {Graphical Models for Processing Missing Data},
journal = {Journal of the American Statistical Association},
volume = {116},
number = {534},
pages = {1023--1037},
year = {2021},
publisher = {Taylor \& Francis},
doi = {10.1080/01621459.2021.1874961}
}

@inproceedings{Mohan_2014,
 author = {Mohan, Karthika and Pearl, Judea},
 booktitle = {Advances in Neural Information Processing Systems},
 editor = {Z. Ghahramani and M. Welling and C. Cortes and N. Lawrence and K. Weinberger},
 pages = {},
 publisher = {Curran Associates, Inc.},
 title = {Graphical Models for Recovering Probabilistic and Causal Queries from Missing Data},
 volume = {27},
 year = {2014}
}

@article{Rubin_1976,
    author = {Rubin, Donald B.},
    title = {Inference and missing data},
    journal = {Biometrika},
    volume = {63},
    number = {3},
    pages = {581-592},
    year = {1976},
    month = {12},
    issn = {0006-3444},
    doi = {10.1093/biomet/63.3.581},
}

@article{Ferreira_2025, 
title={Identifying Macro Conditional Independencies and Macro Total Effects in Summary Causal Graphs with Latent Confounding}, volume={39}, 
DOI={10.1609/aaai.v39i25.34882}, number={25}, journal={Proceedings of the AAAI Conference on Artificial Intelligence}, author={Ferreira, Simon and Assaad, Charles K.}, year={2025}, month={Apr.}, pages={26787-26795} 
}

@inproceedings{Anand_2023,
author = {Anand, Tara V. and Ribeiro, Adele H. and Tian, Jin and Bareinboim, Elias},
title = {Causal effect identification in cluster DAGs},
year = {2023},
isbn = {978-1-57735-880-0},
publisher = {AAAI Press},
doi = {10.1609/aaai.v37i10.26435},
booktitle = {Proceedings of the Thirty-Seventh AAAI Conference on Artificial Intelligence and Thirty-Fifth Conference on Innovative Applications of Artificial Intelligence and Thirteenth Symposium on Educational Advances in Artificial Intelligence},
articleno = {1366},
numpages = {8},
series = {AAAI'23/IAAI'23/EAAI'23}
}

@InProceedings{Assaad_2024,
  title = 	 {Identifiability of total effects from abstractions of time series causal graphs},
  author =       {Assaad, Charles K. and Devijver, Emilie and Gaussier, Eric and Goessler, Gregor and Meynaoui, Anouar},
  booktitle = 	 {Proceedings of the Fortieth Conference on Uncertainty in Artificial Intelligence},
  pages = 	 {173--185},
  year = 	 {2024},
  editor = 	 {Kiyavash, Negar and Mooij, Joris M.},
  volume = 	 {244},
  series = 	 {Proceedings of Machine Learning Research},
  month = 	 {15--19 Jul},
  publisher =    {PMLR},
}

@InProceedings{Assaad_2023,
  title = 	 {Root Cause Identification for Collective Anomalies in Time Series given an Acyclic Summary Causal Graph with Loops},
  author =       {Assaad, Charles K. and Ez-Zejjari, Imad and Zan, Lei},
  booktitle = 	 {Proceedings of The 26th International Conference on Artificial Intelligence and Statistics},
  pages = 	 {8395--8404},
  year = 	 {2023},
  editor = 	 {Ruiz, Francisco and Dy, Jennifer and van de Meent, Jan-Willem},
  volume = 	 {206},
  series = 	 {Proceedings of Machine Learning Research},
  month = 	 {25--27 Apr},
  publisher =    {PMLR}
}

@article{Anand_2025,
  title={Leveraging Cluster Causal Diagrams for Determining Causal Effects in Medicine},
  author={Anand, Tara V. and Hripcsak, George},
  journal={AMIA Annual Symposium Proceedings},
  year={2025},
  note={PMID: 40417561}
}

@inproceedings{Peters_2013,
 author = {Peters, Jonas and Janzing, Dominik and Sch\"{o}lkopf, Bernhard},
 booktitle = {Advances in Neural Information Processing Systems},
 editor = {C.J. Burges and L. Bottou and M. Welling and Z. Ghahramani and K. Weinberger},
 pages = {},
 publisher = {Curran Associates, Inc.},
 title = {Causal Inference on Time Series using Restricted Structural Equation Models},
 volume = {26},
 year = {2013}
}

@article{Inoue_2022,
  title={Causal Effect of Chronic Pain on Mortality Through Opioid Prescriptions: Application of the Front-Door Formula},
  author={Inoue, Keisuke and Ritz, Beate and Arah, Onyebuchi A.},
  journal={Epidemiology},
  year={2022},
  volume={33},
  number={4},
  pages={544--552},
  doi={10.1097/EDE.0000000000001477},
}

@incollection{Wothke_2000,
  author    = {Wothke, Werner},
  title     = {Longitudinal and Multigroup Modeling with Missing Data},
  booktitle = {Modeling Longitudinal and Multilevel Data: Practical Issues, Applied Approaches, and Specific Examples},
  year      = {2000},
  publisher = {Psychology Press},
  doi       = {10.4324/9781410601940}
}

@article{Pereira_2024,
    title = {Imputation of data Missing Not at Random: Artificial generation and benchmark analysis},
    journal = {Expert Systems with Applications},
    volume = {249},
    pages = {123654},
    year = {2024},
    issn = {0957-4174},
    doi = {https://doi.org/10.1016/j.eswa.2024.123654},
    author = {Ricardo Cardoso Pereira and Pedro Henriques Abreu and Pedro Pereira Rodrigues and Mário A.T. Figueiredo}
}

@article{Little_1993,
 ISSN = {01621459, 1537274X},
 author = {Roderick J. A. Little},
 journal = {Journal of the American Statistical Association},
 number = {421},
 pages = {125--134},
 publisher = {[American Statistical Association, Taylor & Francis, Ltd.]},
 title = {Pattern-Mixture Models for Multivariate Incomplete Data},
 urldate = {2026-06-29},
 volume = {88},
 year = {1993}
}

@article{Peek_2018,
    author = {Peek, Niels and Rodrigues, Pedro Pereira},
    year = {2018},
    title = {Three controversies in health data science},
    journal = {International Journal of Data Science and Analytics},
    volume = {6},
    pages = {261-269},
    doi = {https://doi.org/10.1007/s41060-018-0109-y},
    issn = {2364-4168}
}

@article{Horvitz_1952,
  author  = {Horvitz, D. G. and Thompson, D. J.},
  title   = {A Generalization of Sampling Without Replacement from a Finite Universe},
  journal = {Journal of the American Statistical Association},
  volume  = {47},
  number  = {260},
  pages   = {663--685},
  year    = {1952},
  doi     = {10.1080/01621459.1952.10483446}
}

@article{Robins_1995,
  author  = {Robins, James M. and Rotnitzky, Andrea and Zhao, Lue Ping},
  title   = {Analysis of Semiparametric Regression Models for Repeated Outcomes in the Presence of Missing Data},
  journal = {Journal of the American Statistical Association},
  volume  = {90},
  number  = {429},
  pages   = {106--121},
  year    = {1995},
  doi     = {10.1080/01621459.1995.10476493}
}

@article{Gibbs_2002,
  author  = {Gibbs, Alison L. and Su, Francis Edward},
  title   = {On Choosing and Bounding Probability Metrics},
  journal = {International Statistical Review},
  volume  = {70},
  number  = {3},
  pages   = {419--435},
  year    = {2002},
  doi     = {10.1111/j.1751-5823.2002.tb00178.x}
}

@misc{WHO_HIV_2025,
  author  = {{World Health Organization}},
  title   = {HIV and AIDS},
  date    = {2025-07-15},
  url     = {https://www.who.int/news-room/fact-sheets/detail/hiv-aids},
  urldate = {2026-07-24}
}

@article{Mustanski_2007,
	title = {Psychosocial health problems increase risk for {HIV} among urban young men who have sex with men: Preliminary evidence of a syndemic in need of attention},
	volume = {34},
	issn = {1532-4796},
	doi = {10.1007/BF02879919},
	shorttitle = {Psychosocial health problems increase risk for {HIV} among urban young men who have sex with men},
	pages = {37--45},
	number = {1},
	journaltitle = {Annals of Behavioral Medicine},
	author = {Mustanski, Brian and Garofalo, Robert and Herrick, Amy and Donenberg, Geri},
	urldate = {2026-07-24},
	date = {2007},
	langid = {english}
}

@article{Alvy_2011,
	title = {Depression is Associated with Sexual Risk Among Men Who Have Sex with Men, but is Mediated by Cognitive Escape and Self-Efficacy},
	volume = {15},
	issn = {1573-3254},
	doi = {10.1007/s10461-010-9678-z},
	pages = {1171--1179},
	number = {6},
	journaltitle = {{AIDS} and Behavior},
	author = {Alvy, Lisa M. and McKirnan, David J. and Mansergh, Gordon and Koblin, Beryl and Colfax, Grant N. and Flores, Stephen A. and Hudson, Sharon and {Project MIX Study Group}},
	urldate = {2026-07-24},
	date = {2011-08-01},
	langid = {english}
}

@article{Ocleirigh_2013,
	title = {Moderate Levels of Depression Predict Sexual Transmission Risk in {HIV}-Infected {MSM}: A Longitudinal Analysis of Data From Six Sites Involved in a “Prevention for Positives” Study},
	volume = {17},
	issn = {1573-3254},
	doi = {10.1007/s10461-013-0462-8},
	shorttitle = {Moderate Levels of Depression Predict Sexual Transmission Risk in {HIV}-Infected {MSM}},
	pages = {1764--1769},
	number = {5},
	journaltitle = {{AIDS} and Behavior},
	author = {O’Cleirigh, Conall and Newcomb, Michael E. and Mayer, Kenneth H. and Skeer, Margie and Traeger, Lara and Safren, Steven A.},
	urldate = {2026-07-24},
	date = {2013},
	langid = {english}
}

@article{Tikka_2023,
author = {Tikka, Santtu and Helske, Jouni and Karvanen, Juha},
title = {Clustering and structural robustness in causal diagrams},
year = {2023},
issue_date = {January 2023},
publisher = {JMLR.org},
volume = {24},
number = {1},
issn = {1532-4435},
journal = {J. Mach. Learn. Res.},
month = jan,
articleno = {195},
numpages = {32}
}


\newpage

\appendix

\section{Proofs}

\subsection{Proof of Proposition~\ref{proposition:compatibility_class}}
\begin{proof}
Let \(\mathcal{G}^m \in \compatible{\mathcal{G}^{\mathbb{C}, m}}\). We show that
\(\mathcal{G} \in \compatible{\mathcal{G}^{\mathbb{C}, cm}}\), where
\(\mathcal{G}^{\mathbb{C}, cm}=\pi(\mathcal{G}^{\mathbb{C}, m})\).

By definition, \(\pi\) is obtained by merging, for each cluster \(C_{\mathbb{X}}\in\mathbb{C}\), all variable-level missingness indicators
\[
\{R_X : X\in C_{\mathbb{X}}\}
\]
into a single cluster-level missingness indicator \(R_{C_{\mathbb{X}}}\). All cluster-level adjacencies and orientations induced by the substantive variables are preserved, while distinctions between missingness indicators inside the same cluster are forgotten.

Since \(\mathcal{G}^m\) is compatible with \(\mathcal{G}^{\mathbb{C}, m}\), every variable-level edge in \(\mathcal{G}^m\) induces an edge allowed by \(\mathcal{G}^{\mathbb{C}, m}\). After applying \(\pi\), any edge involving a variable-level missingness indicator \(R_X\), with \(X\in C_{\mathbb{X}}\), is mapped to the corresponding edge involving \(R_{C_{\mathbb{X}}}\). Thus, every edge of \(\mathcal{G}^m\) also induces an edge allowed by \(\mathcal{G}^{\mathbb{C}, cm}\). Therefore,
\[
\mathcal{G}^m \in \compatible{\mathcal{G}^{\mathbb{C}, cm}}.
\]
As this holds for every \(\mathcal{G} \in \compatible{\mathcal{G}^m}\), we obtain
\[
\compatible{\mathcal{G}^{\mathbb{C}, m}}
\subseteq
\compatible{\mathcal{G}^{\mathbb{C}, cm}}.
\]

It remains to show that the inclusion is strict in general. Consider a cluster
\(C_{\mathbb{X}}=\{X_1,X_2\}\). In an m-C-DMG, the missingness mechanisms of \(X_1\) and \(X_2\) are represented by two distinct vertices, \(R_{X_1}\) and \(R_{X_2}\). Therefore, the graph can distinguish, for instance, a variable-level m-ADMG in which some variable \(L\) points to \(R_{X_1}\) but not to \(R_{X_2}\).

After applying \(\pi\), both \(R_{X_1}\) and \(R_{X_2}\) are merged into a single vertex \(R_{C_{\mathbb{X}}}\). The cm-C-DMG only records that \(L\) may affect the missingness mechanism of the cluster \(C_{\mathbb{X}}\), without specifying whether this effect is on \(R_{X_1}\), on \(R_{X_2}\), or on both. Hence a variable-level m-ADMG in which \(L\) points to \(R_{X_2}\) but not to \(R_{X_1}\) is compatible with the cm-C-DMG, but not with the original m-C-DMG if the latter allowed only the edge \(L \to R_{X_1}\).

Thus, there exist variable-level m-ADMGs compatible with
\(\mathcal{G}^{cm}\) that are not compatible with \(\mathcal{G}^m\). Therefore, the inclusion is strict in general:
\[
\compatible{\mathcal{G}^{\mathbb{C}, m}}
\subsetneq
\compatible{\mathcal{G}^{\mathbb{C}, cm}}.
\]

\end{proof}

\subsection{Proof of Proposition~\ref{proposition:compatibility_recover}}

\begin{proof}
By Proposition~\ref{proposition:compatibility_class},
$\compatible{\mathcal{G}^{m}_{\mathbb{C}}}
\subseteq
\compatible{\mathcal{G}^{\mathbb{C}, cm}}$.
Assume that \(Q\) is recoverable in
\(\mathcal{G}^{\mathbb{C}, cm}\). By definition, this means that \(Q\) is
uniquely determined from the observed distribution for all variable-level
m-ADMGs in
\(\compatible{\mathcal{G}^{\mathbb{C}, cm}}\). Since
\(\compatible{\mathcal{G}^{m}_{\mathbb{C}}}\) is a subset of this compatibility
class, the same uniqueness holds for all variable-level m-ADMGs compatible with
\(\mathcal{G}^{m}_{\mathbb{C}}\). Therefore, \(Q\) is recoverable in
\(\mathcal{G}^{m}_{\mathbb{C}}\).

We now show that the converse does not hold in general. Let
\(C=\{A,B\}\) be a cluster and consider the target query
\[
Q=P(A).
\]
Suppose that, in the m-C-DMG, the cluster \(C\) points to the missingness
indicator \(R_B\), but not to \(R_A\):
\[
C \to R_B,
\qquad
C \not\to R_A.
\]
Thus, the m-C-DMG distinguishes the missingness mechanism of \(A\) from that of
\(B\). In particular, it records that the missingness of \(A\) is not affected
by variables in \(C\). Under this abstraction, \(P(A)\) is recoverable from the
observed distribution.

In the associated cm-C-DMG, however, the indicators \(R_A\) and \(R_B\) are
merged into a single vertex \(R_C\). The edge \(C\to R_B\) is therefore abstracted
as
\[
C \to R_C.
\]
This coarser graph no longer specifies whether the missingness mechanism
affected by \(C\) is that of \(A\), that of \(B\), or both. Hence it is
compatible with variable-level m-ADMGs in which the missingness of \(A\) depends
on \(A\) itself. In such cases, \(P(A)\) is generally not recoverable because of
self-censoring.

Therefore, \(P(A)\) can be recoverable in the finer m-C-DMG but not in the
associated cm-C-DMG. Hence the converse implication fails in general.
\end{proof}

\subsection{Proof of Theorem~\ref{theorem:necessary_recovering_joint_dist}}

We first prove the result for cm-C-DMGs.

\begin{definition}[Primary path~\citep{Ferreira_2025}]
\label{def:primary_path}

Let $\mathcal{G}^*=(\mathbb{V}^*,\mathbb{E}^*)$ be a graph and let

$
\widetilde{\pi}
=
\langle
V^*_1,\ldots,V^*_n
\rangle
$
be a walk.
The primary path of $\widetilde{\pi}$, denoted by
$\widetilde{\pi}_p$, is defined iteratively as follows. Set
$
U^*_1=V^*_1.
$
For every $k\geq 1$, define
$
U^*_{k+1}
=
V^*_{\max{i:V^*_i=U^*_k}+1},
$
until the last vertex $V^*_n$ is reached.
The edge between $U^*_k$ and $U^*_{k+1}$ in
$\widetilde{\pi}_p$ has the same endpoint marks as the edge between
$V^*_{\max{i:V^*_i=U^*_k}}
$
and
$
V^*_{\max{i:V^*_i=U^*_k}+1}
$
in $\widetilde{\pi}$.
\end{definition}

\begin{lemma}
\label{lem:collider_path_equivalence}

Let $\mathcal{G}^{\mathbb{C},cm}$ be a cm-C-DMG. For every
$C_{\mathbb{X}}\in\mathbb{C}^m$, there exists a collider path between
$C_{\mathbb{X}}$ and $R_{C_{\mathbb{X}}}$ in
$\mathcal{G}^{\mathbb{C},cm}$, with every intermediate vertex in
$\mathbb{C}^m\cup\mathbb{C}^o$, if and only if there exist a compatible
variable-level m-ADMG $\mathcal{G}^m$ and some
$X\in C_{\mathbb{X}}\cap \mathbb V^m$ such that $\mathcal{G}^m$ contains a path between
$X$ and $R_X$ whose intermediate vertices are substantive colliders.
\end{lemma}

\begin{proof} 
Let $\rho(V)=C_V$ denote the clustering map for substantive variables, and let $\rho(R_V)=R_{C_V}$ for their missingness indicators. By compatibility, a variable-level edge between vertices in two distinct clusters induces an edge between the corresponding cluster vertices with the same endpoint marks. Conversely, in every compatible variable-level graph, each
cluster-level edge has at least one variable-level witness between its
endpoint clusters with the same endpoint marks. We prove the two implications separately. 

\noindent $(\Leftarrow)$ Suppose that there exist a variable-level m-ADMG $\mathcal{G}^m$ compatible with $\mathcal{G}^{\mathbb{C},cm}$ and some $X\in C_{\mathbb{X}}\cap \mathbb V^m$ such that $\mathcal{G}^m$ contains a path \[ \pi = \langle X,V_1,\ldots,V_k,R_X \rangle \] whose intermediate vertices $V_1,\ldots,V_k$ are substantive colliders. 

Apply $\rho$ to the vertices of $\pi$ and contract every maximal consecutive subsequence of identical cluster vertices to one occurrence. Denote the resulting cluster-level sequence by
\[ \rho(\pi) = \langle D_0=C_{\mathbb{X}}, D_1,\ldots,D_{m-1}, D_m=R_{C_{\mathbb{X}}} \rangle \] in $\mathcal{G}^{\mathbb{C},cm}$. This sequence is a walk rather than necessarily a path because distinct variable-level vertices may belong to the same cluster. Let \[ \rho(\pi)_p = \langle U_0,U_1,\ldots,U_\ell \rangle \] be the primary path associated with $\rho(\pi)$. Its endpoints are 
$U_0=C_{\mathbb{X}}$ and $U_\ell=R_{C_{\mathbb{X}}}.$

We show that every intermediate vertex $U_j$, $1\leq j\leq \ell-1$, is a collider. Fix such a vertex $U_j$. Let $s_j$ and $t_j$ denote, respectively, the first and last occurrences of $U_j$ in the portion of $\rho(\pi)$ retained by the primary-path construction. The edge between $U_{j-1}$ and $U_j$ in the primary path is inherited from the edge entering the first occurrence of $U_j$, whereas the edge between $U_j$ and $U_{j+1}$ is inherited from the edge leaving the last occurrence of $U_j$. Let $V_{s_j}$ and $V_{t_j}$ be the corresponding variable-level vertices on $\pi$. Since $U_j$ is an intermediate vertex of the primary path, $V_{s_j}$ and $V_{t_j}$ are intermediate substantive vertices of $\pi$. By assumption, every intermediate vertex of $\pi$ is a collider. Hence, the edge entering $V_{s_j}$ has an arrowhead at $V_{s_j}$, and the edge leaving $V_{t_j}$ has an arrowhead at $V_{t_j}$. Compatibility preserves endpoint marks under clustering. Therefore, the edge entering $U_j$ in $\rho(\pi)_p$ has an arrowhead at $U_j$, and the edge leaving $U_j$ also has an arrowhead at $U_j$. Thus, locally, 
$$U_{j-1} *\!\rightarrow U_j \leftarrow\!* U_{j+1},$$
where $*$ denotes either a tail or an arrowhead at the other endpoint. Consequently, $U_j$ is a collider. This argument applies to every intermediate vertex of $\rho(\pi)_p$. Moreover, each such vertex is the cluster of  substantive variables and therefore belongs to $\mathbb{C}^m\cup\mathbb{C}^o$. Hence, $\rho(\pi)_p$ is a path between $C_{\mathbb{X}}$ and $R_{C_{\mathbb{X}}}$ whose intermediate vertices are substantive cluster-level colliders. 

\noindent $(\Rightarrow)$
Suppose that $\mathcal G^{\mathbb C,cm}$ contains a path
$$
\pi_{\mathbb C}
=
\langle
C_0,C_1,\ldots,C_k,C_{k+1}
\rangle,
$$
where
$$
C_0=C_{\mathbb X},
\qquad
C_{k+1}=R_{C_{\mathbb X}},
$$
and every intermediate vertex
$$
C_i\in\mathbb C^o\cup\mathbb C^m,
\qquad
1\leq i\leq k,
$$
is a collider on $\pi_{\mathbb C}$.

Choose a partially observed variable
$
X\in C_{\mathbb X}\cap\mathbb V^m,
$
and let $R_X\in R_{C_{\mathbb X}}$ be its missingness indicator.

Because $\mathcal G^{\mathbb C,cm}$ admits at least one compatible
variable-level m-ADMG, choose
$
\mathcal G_0^m
\in
\compatible{\mathcal G^{\mathbb C,cm}}.
$

We modify $\mathcal G_0^m$ so that it contains a variable-level
counterpart of $\pi_{\mathbb C}$.

\paragraph{The case $k=0$}

In this case, $C_{\mathbb X}$ and $R_{C_{\mathbb X}}$ are adjacent.
Because a missingness indicator has no substantive children, their edge
is either
$
C_{\mathbb X}\rightarrow R_{C_{\mathbb X}}
$
or
$
C_{\mathbb X}\longdashleftrightarrow R_{C_{\mathbb X}}.
$
In the first case, add
$
X\rightarrow R_X
$
to $\mathcal G_0^m$ if it is not already present. 
This cannot create a directed cycle because no directed path from $R_X$
can reach the substantive variable $X$, since missingness indicators are
not parents of substantive or latent variables.

In the second case, add
$
X\longdashleftrightarrow R_X
$
if it is not already present. This does not alter the directed part of
the graph.

In either case, the added edge projects to an edge already present in
$\mathcal G^{\mathbb C,cm}$. Hence, the resulting graph remains
compatible and contains the required length-one path between $X$ and
$R_X$.

\paragraph{The case $k\geq 1$}
Because $C_1$ is a collider on $\pi_{\mathbb C}$, the first edge has an
arrowhead at $C_1$. Therefore, it is either
$
C_0\rightarrow C_1
$
or
$
C_0\longdashleftrightarrow C_1.
$

We construct the first variable-level edge separately in these two
cases.

Suppose first that
$
C_0\rightarrow C_1.
$
Since $\mathcal G_0^m$ is compatible with
$\mathcal G^{\mathbb C,cm}$, there exist substantive variables
$$
A\in C_0
\qquad\text{and}\qquad
B\in C_1
$$
such that
$
A\rightarrow B
$
is an edge of $\mathcal G_0^m$.

Let $\sigma$ be the permutation of the variable-level vertices that
exchanges $A$ and $X$ and fixes every other vertex. Let
$\mathcal G_1^m$ be obtained by applying $\sigma$ to both endpoints of
every edge of $\mathcal G_0^m$.

The graphs $\mathcal G_0^m$ and $\mathcal G_1^m$ are isomorphic, so the
directed part of $\mathcal G_1^m$ is acyclic. Moreover, $A$ and $X$
belong to the same cluster. Consequently, applying $\sigma$ does not
change the cluster-level projection of any edge, and therefore
$
\mathcal G_1^m
\in
\compatible{\mathcal G^{\mathbb C,cm}}.
$
The edge $A\rightarrow B$ is transformed into
$
X\rightarrow B.
$
Set
$
V_1=B.
$

Suppose instead that
$
C_0\longdashleftrightarrow C_1.
$
Set
$
\mathcal G_1^m=\mathcal G_0^m,
$
choose any substantive variable
$
V_1\in C_1,
$
and add
$
X\longdashleftrightarrow V_1
$
if it is not already present. Since bows are permitted, this addition is
valid even if a directed edge already joins $X$ and $V_1$. The added
bidirected edge does not alter the directed part of the graph and
projects to the existing cluster-level edge
$C_0\longdashleftrightarrow C_1$.

We now construct the remaining edges of the path.

For every $1\leq i\leq k-1$, both $C_i$ and $C_{i+1}$ are colliders on
$\pi_{\mathbb C}$. Hence, the edge between them must have arrowheads at
both endpoints and is therefore bidirected:
$$
C_i\longdashleftrightarrow C_{i+1}.
$$
Recursively choose any substantive variable
$$
V_{i+1}\in C_{i+1}
$$
and add the bidirected edge
$$
V_i\longdashleftrightarrow V_{i+1}
$$
if it is not already present.

Finally, consider the edge between $C_k$ and
$R_{C_{\mathbb X}}$. Because $C_k$ is a collider, this edge has an
arrowhead at $C_k$. It cannot be of the form
$$
R_{C_{\mathbb X}}\rightarrow C_k,
$$
because missingness indicators have no substantive children. It also
cannot be of the form
$$
C_k\rightarrow R_{C_{\mathbb X}},
$$
because that edge has a tail at $C_k$. Therefore, the final edge must be
$$
C_k\longdashleftrightarrow R_{C_{\mathbb X}}.
$$
Add
$$
V_k\longdashleftrightarrow R_X
$$
if it is not already present.

Let $\mathcal G^m$ be the graph obtained from $\mathcal G_1^m$ after
these additions. Every added edge is bidirected, so the directed part of
$\mathcal G^m$ is identical to the directed part of
$\mathcal G_1^m$ and is therefore acyclic. The only added edge incident
to a missingness indicator has an arrowhead at that indicator, so no
substantive child of a missingness indicator is introduced. Since bows
are permitted, the additions preserve the m-ADMG restrictions.

Furthermore, each added edge projects to an edge already present on
$\pi_{\mathbb C}$. Thus, no new cluster-level edge is introduced.
Because $\mathcal G_1^m$ was already compatible with
$\mathcal G^{\mathbb C,cm}$, it follows that
$$
\mathcal G^m
\in
\compatible{\mathcal G^{\mathbb C,cm}}.
$$

By construction, $\mathcal G^m$ contains the path
$$
\pi
=
\langle
X,V_1,\ldots,V_k,R_X
\rangle.
$$

The first edge has an arrowhead at $V_1$, every edge between consecutive
intermediate vertices is bidirected, and the final edge has an arrowhead
at $V_k$. Consequently, for every $1\leq i\leq k$, both edges incident
to $V_i$ on $\pi$ have arrowheads at $V_i$. Hence, every intermediate
vertex $V_i$ is a substantive collider.

Therefore, there exist a compatible variable-level m-ADMG
$\mathcal G^m$ and a partially observed variable
$$
X\in C_{\mathbb X}\cap\mathbb V^m
$$
such that $\mathcal G^m$ contains a path between $X$ and $R_X$ whose
intermediate vertices are substantive colliders.
\end{proof}

We now prove the theorem.
\begin{proof}
    Let
$
\mathcal{G}^{\mathbb{C},cm}
=
\left(
\mathbb{C}^o
\cup
\mathbb{C}^m
\cup
\mathbb{R}_{\mathbb{C}},
\mathbb{E}^{\mathbb{C}}
\right)
$
be a cm-C-DMG. 
Suppose first that there exists
$C_{\mathbb{X}}\in\mathbb{C}^m$ such that
$C_{\mathbb{X}}$ and $R_{C_{\mathbb{X}}}$ are neighbors.

By the definition of compatibility, we may construct a compatible
variable-level m-ADMG $\mathcal{G}^m$ containing some
$X\in C_{\mathbb{X}}$ and its missingness indicator
$R_X\in R_{C_{\mathbb{X}}}$ such that $X$ and $R_X$ are adjacent.

The variable-level necessary condition of \citet{Mohan_2014} is violated
in this compatible graph. Therefore, the joint distribution is not
recoverable in $\mathcal{G}^m$. Since the cm-C-DMG represents a class
containing this compatible graph, the joint distribution cannot be
recoverable uniformly over the class represented by
$\mathcal{G}^{\mathbb{C},cm}$.

Suppose next that there exists
$C_{\mathbb{X}}\in\mathbb{C}^m$ such that
$C_{\mathbb{X}}$ and $R_{C_{\mathbb{X}}}$ are connected by a path whose
intermediate vertices are colliders belonging to
$\mathbb{C}^m\cup\mathbb{C}^o$.

By Lemma~\ref{lem:collider_path_equivalence}, there exist a compatible
variable-level m-ADMG $\mathcal{G}^m$ and some
$X\in C_{\mathbb{X}}$ such that $X$ and $R_X$ are connected by a path
whose intermediate vertices are substantive colliders.

The second variable-level necessary condition of
\citet{Mohan_2014} is therefore violated. Consequently, the joint
distribution is not recoverable in $\mathcal{G}^m$ and cannot be
recoverable uniformly over all variable-level m-ADMGs compatible with
$\mathcal{G}^{\mathbb{C},cm}$.

Thus, the following conditions are necessary for recoverability:

\begin{enumerate}
\item for every $C_{\mathbb{X}}\in\mathbb{C}^m$,
$C_{\mathbb{X}}$ and $R_{C_{\mathbb{X}}}$ are not neighbors;
\item for every $C_{\mathbb{X}}\in\mathbb{C}^m$, there is no path
between $C_{\mathbb{X}}$ and $R_{C_{\mathbb{X}}}$ whose
intermediate vertices are colliders belonging to
$\mathbb{C}^m\cup\mathbb{C}^o$.
\end{enumerate}

\paragraph{Extension to m-C-DMGs}

The proof for m-C-DMGs follows the same argument. It follows by applying the same strategy  to the missingness-indicator
vertices of the m-C-DMG, 
Therefore, the stated graphical conditions are necessary 
for recovering the joint distribution in both cm-C-DMGs and m-C-DMGs.
\end{proof}

\subsection{Proof of Theorem~\ref{theorem:sufficient_recovering_joint_dist}}
\begin{definition}[Missingness-indicator blocks]
\label{def:missingness_indicator_blocks}

Let
$
\mathcal{G}^{\mathbb{C},cm}
=
\left(
\mathbb{C}^o
\cup
\mathbb{C}^m
\cup
\mathbb{R}_{\mathbb{C}},
\mathbb{E}^{\mathbb{C}}
\right)
$
be a cm-C-DMG.
Construct an undirected auxiliary graph
$\mathcal{H}_{\mathbb{R}}$ whose vertex set is
$\mathbb{R}_{\mathbb{C}}$. Two distinct missingness-indicator vertices
$R_{C_i}$ and $R_{C_j}$ are adjacent in
$\mathcal{H}_{\mathbb{R}}$ if there exists a collider path between them
in $\mathcal{G}^{\mathbb{C},cm}$.
Let
$
\mathbb{R}^{(1)},\ldots,\mathbb{R}^{(K)}
$
be the connected components of
$\mathcal{H}_{\mathbb{R}}$. We call these sets the
missingness-indicator blocks of
$\mathcal{G}^{\mathbb{C},cm}$.
For every block $\mathbb{R}^{(a)}$, define
$$
\mathbb{M}_a
=
\left\{
C\in\mathbb{C}^o\cup\mathbb{C}^m:
\begin{array}{l}
\text{there exist }R_{C_i}\in\mathbb{R}^{(a)}
\text{ and a collider path}\\
\text{between }R_{C_i}\text{ and }C
\text{ in }\mathcal{G}^{\mathbb{C},cm}
\end{array}
\right\}.
$$
Paths of length one are included. Thus, every substantive cluster
adjacent to a member of $\mathbb{R}^{(a)}$ belongs to $\mathbb{M}_a$.
\end{definition}

\begin{lemma}[Cluster Markov blanket of a missingness block]
\label{lem:cluster_markov_blanket}

Let $\mathbb{R}^{(a)}$ be a missingness-indicator block and let
$\mathbb{M}_a$ be defined as in
Definition~\ref{def:missingness_indicator_blocks}. If there are no edges between  missingness-indicator
vertices
then
$$
\indepc{
\mathbb{R}^{(a)}
}{
\left(
\mathbb{C}^o
\cup
\mathbb{C}^m
\cup
\mathbb{R}_{\mathbb{C}}
\right)
\setminus
\left(
\mathbb{R}^{(a)}
\cup
\mathbb{M}_a
\right)
}{
\mathbb{M}_a
}{
\mathcal{G}^{\mathbb{C},cm}
}.
$$

Consequently, $\mathbb{M}_a$ is a substantive cluster Markov blanket of
the missingness-indicator block $\mathbb{R}^{(a)}$.
\end{lemma}

\begin{proof}
Fix arbitrary vertices
$
R_{C_i}\in\mathbb{R}^{(a)}
$
and
$
Z
\in
\left(
\mathbb{C}^o
\cup
\mathbb{C}^m
\cup
\mathbb{R}_{\mathbb{C}}
\right)
\setminus
\left(
\mathbb{R}^{(a)}
\cup
\mathbb{M}_a
\right).
$
Consider an arbitrary path
$$
\pi
=
\langle
R_{C_i}=Z_0,Z_1,\ldots,Z_q=Z
\rangle
$$
between $R_{C_i}$ and $Z$.
We show that $\pi$ is blocked by $\mathbb{M}_a$.

Suppose first that every intermediate vertex
$Z_1,\ldots,Z_{q-1}$ is a collider. Then $\pi$ is a collider path.

If $Z$ is a substantive cluster, then the definition of
$\mathbb{M}_a$ implies
$
Z\in\mathbb{M}_a,
$
contradicting the choice of $Z$.

If $Z$ is a missingness-indicator vertex, then
$R_{C_i}$ and $Z$ are adjacent in the auxiliary graph
$\mathcal{H}_{\mathbb{R}}$. Hence, they belong to the same connected
component, and therefore
$
Z\in\mathbb{R}^{(a)},
$
again contradicting the choice of $Z$.

It follows that $\pi$ contains at least one intermediate noncollider.

Let $Z_j$ (such that $1\leq j\leq q-1$) be the first noncollider encountered when traversing $\pi$
from $R_{C_i}$ toward $Z$. Thus,
$$
Z_1,\ldots,Z_{j-1}
$$
are colliders and $Z_j$ is a noncollider.

We first show that $Z_j$ is a substantive cluster.

Suppose, toward a contradiction, that $Z_j\in\mathbb R_{\mathbb C}$. Since $\pi$ is a path and there are no edges between distinct missingness-indicator vertices, both $Z_{j-1}$ and $Z_{j+1}$ are substantive clusters. Moreover, missingness-indicator vertices have no substantive children. Therefore, neither incident edge can have a tail at $Z_j$. Both edges incident to $Z_j$ on $\pi$ consequently have arrowheads at $Z_j$. Thus, $Z_j$ is a collider, contradicting its definition.
Hence,
$
Z_j\in\mathbb{C}^o\cup\mathbb{C}^m.
$
Since $Z_1,\ldots,Z_{j-1}$ are colliders, the subpath
$$
\langle
R_{C_i},Z_1,\ldots,Z_j
\rangle
$$
is a collider path between a member of $\mathbb{R}^{(a)}$ and the
substantive cluster $Z_j$. By the definition of $\mathbb{M}_a$,
$
Z_j\in\mathbb{M}_a.
$

The vertex $Z_j$ is a noncollider on $\pi$ and belongs to the
conditioning set $\mathbb{M}_a$. Therefore, $\pi$ is blocked by
$\mathbb{M}_a$.

Since the argument applies to every path from every member of
$\mathbb{R}^{(a)}$ to every vertex outside
$\mathbb{R}^{(a)}\cup\mathbb{M}_a$, we obtain
$$
\indepc{
\mathbb{R}^{(a)}
}{
\left(
\mathbb{C}^o
\cup
\mathbb{C}^m
\cup
\mathbb{R}_{\mathbb{C}}
\right)
\setminus
\left(
\mathbb{R}^{(a)}
\cup
\mathbb{M}_a
\right)
}{
\mathbb{M}_a
}{
\mathcal{G}^{\mathbb{C},cm}
}.
$$
\end{proof}

\begin{lemma}[Transfer of the cluster Markov blanket]
\label{lem:cluster_markov_blanket_transfer}

Let $\mathbb{R}^{(a)}$ be a missingness-indicator block and let
$\mathbb{M}_a$ be its substantive cluster Markov blanket as defined
above.
For a collection $\mathbb{A}$ of cluster vertices, let
$
\operatorname{var}(\mathbb{A})
$
denote the union of all variable-level vertices represented by the
clusters in $\mathbb{A}$.
Define
$$
\mathbf{R}^{(a)}
=
\operatorname{var}
\left(
\mathbb{R}^{(a)}
\right)
$$
and
$$
\mathbf{M}_a
=
\operatorname{var}
\left(
\mathbb{M}_a
\right).
$$
Then, in every variable-level m-ADMG $\mathcal{G}^m$ compatible with
$\mathcal{G}^{\mathbb{C},cm}$,
$$
\indepc{
\mathbf{R}^{(a)}
}{
\operatorname{var}
\left(
\left(
\mathbb{C}^o
\cup
\mathbb{C}^m
\cup
\mathbb{R}_{\mathbb{C}}
\right)
\setminus
\left(
\mathbb{R}^{(a)}
\cup
\mathbb{M}_a
\right)
\right)
}{
\mathbf{M}_a
}{
\mathcal{G}^m
}.
$$
Furthermore, if $P$ satisfies the global Markov property with respect to
$\mathcal{G}^m$, then
$$
\indepc{
\mathbf{R}^{(a)}}
{
\operatorname{var}
\left(
\left(
\mathbb{C}^o
\cup
\mathbb{C}^m
\cup
\mathbb{R}_{\mathbb{C}}
\right)
\setminus
\left(
\mathbb{R}^{(a)}
\cup
\mathbb{M}_a
\right)
\right)
}
{
\mathbf{M}_a}{P}.
$$
Consequently, $\mathbf{M}_a$ is a variable-level Markov blanket of the
missingness-indicator block $\mathbf{R}^{(a)}$ in every compatible
variable-level m-ADMG.
\end{lemma}

\begin{proof}

By Lemma~\ref{lem:cluster_markov_blanket},
$$
\indepc{
\mathbb{R}^{(a)}
}{
\left(
\mathbb{C}^o
\cup
\mathbb{C}^m
\cup
\mathbb{R}_{\mathbb{C}}
\right)
\setminus
\left(
\mathbb{R}^{(a)}
\cup
\mathbb{M}_a
\right)
}{
\mathbb{M}_a
}{
\mathcal{G}^{\mathbb{C},cm}
}.
$$
Notice that the sets
$$
\mathbb R^{(a)},
\qquad
\left(
\mathbb C^o\cup\mathbb C^m\cup\mathbb R_{\mathbb C}
\right)
\setminus
\left(
\mathbb R^{(a)}\cup\mathbb M_a
\right),
\qquad
\mathbb M_a
$$
are pairwise disjoint. 

The soundness of cluster d-separation~\citep{Ferreira_2025} states that if three disjoint collections $\mathbb{A},\mathbb{B},\mathbb{D}$
of cluster vertices are d-separated in the cluster graph, then the
unions of the variable-level vertices represented by those collections
are d-separated in every compatible variable-level graph.
Applying soundness with
$$
\mathbb{A}
=
\mathbb{R}^{(a)},
$$

$$
\mathbb{B}
=
\left(
\mathbb{C}^o
\cup
\mathbb{C}^m
\cup
\mathbb{R}_{\mathbb{C}}
\right)
\setminus
\left(
\mathbb{R}^{(a)}
\cup
\mathbb{M}_a
\right),
$$
and
$$
\mathbb{D}
=
\mathbb{M}_a
$$
gives, in every compatible variable-level m-ADMG
$\mathcal{G}^m$,
$$
\indepc{
\operatorname{var}
\left(
\mathbb{R}^{(a)}
\right)
}{
\operatorname{var}
\left(
\left(
\mathbb{C}^o
\cup
\mathbb{C}^m
\cup
\mathbb{R}_{\mathbb{C}}
\right)
\setminus
\left(
\mathbb{R}^{(a)}
\cup
\mathbb{M}_a
\right)
\right)
}{
\operatorname{var}
\left(
\mathbb{M}_a
\right)
}{
\mathcal{G}^m
}.
$$
Using the notation
$$
\mathbf{R}^{(a)}
=
\operatorname{var}
\left(
\mathbb{R}^{(a)}
\right)
$$
and
$$
\mathbf{M}_a
=
\operatorname{var}
\left(
\mathbb{M}_a
\right),
$$
this becomes
$$
\indepc{
\mathbf{R}^{(a)}
}{
\operatorname{var}
\left(
\left(
\mathbb{C}^o
\cup
\mathbb{C}^m
\cup
\mathbb{R}_{\mathbb{C}}
\right)
\setminus
\left(
\mathbb{R}^{(a)}
\cup
\mathbb{M}_a
\right)
\right)
}{
\mathbf{M}_a
}{
\mathcal{G}^m
}.
$$

Now fix an arbitrary distribution $P$ satisfying the global Markov
property with respect to $\mathcal{G}^m$. The preceding graphical
d-separation entails the probabilistic conditional independence

$$
\indepc{
\mathbf{R}^{(a)}
}
{
\operatorname{var}
\left(
\left(
\mathbb{C}^o
\cup
\mathbb{C}^m
\cup
\mathbb{R}_{\mathbb{C}}
\right)
\setminus
\left(
\mathbb{R}^{(a)}
\cup
\mathbb{M}_a
\right)
\right)
}
{
\mathbf{M}_a}{P}.
$$

Since $\mathcal{G}^m$ was arbitrary, the conclusion holds in every
variable-level m-ADMG compatible with
$\mathcal{G}^{\mathbb{C},cm}$.
\end{proof}

We now prove the theorem.
\begin{proof}

Let $\mathcal G^m$ be an arbitrary variable-level m-ADMG compatible
with $\mathcal G^{\mathbb C,cm}$, and let $P$ be an arbitrary
distribution satisfying the global Markov property with respect to
$\mathcal G^m$.

For every missingness-indicator block $\mathbb R^{(a)}$, let
$$
\mathbf R^{(a)}
=
\operatorname{var}
\left(
\mathbb R^{(a)}
\right)
$$
and
$$
\mathbf M_a
=
\operatorname{var}
\left(
MB^o\left(\mathbb R^{(a)}\right)
\cup
MB^m\left(\mathbb R^{(a)}\right)
\right).
$$
Write
$$
\mathbf M_a^o
=
\operatorname{var}
\left(
MB^o\left(\mathbb R^{(a)}\right)
\right)
$$
and
$$
\mathbf M_a^m
=
\operatorname{var}
\left(
MB^m\left(\mathbb R^{(a)}\right)
\right).
$$

By Lemma~\ref{lem:cluster_markov_blanket_transfer},
\begin{equation}
\label{eq:r_indep_given_mb}    
\indepc{\mathbf R^{(a)}}
{\left[
\left(
\mathbb V^o\cup\mathbb V^m
\right)
\setminus
\mathbf M_a
\right]
\cup
\left[
\mathbb R\setminus\mathbf R^{(a)}
\right]}
{\mathbf M_a}{P}.
\end{equation}

Choose an arbitrary ordering

$$
\mathbb R^{(1)},\ldots,\mathbb R^{(K)}
$$

of the missingness-indicator blocks. By the chain rule,

\begin{equation}
\label{eq:r_factorization}    
P
\left(
\mathbb R
\mid
\mathbb V^o,\mathbb V^m
\right)
=
\prod_{a=1}^{K}
P\left(
\mathbf R^{(a)}
\mid
\mathbb V^o,
\mathbb V^m,
\mathbf R^{(1)},\ldots,\mathbf R^{(a-1)}
\right).
\end{equation}

For each $a$, all variables in the conditioning set of the corresponding
factor in \eqref{eq:r_factorization}, except those in $\mathbf M_a$, belong to the set on the
right-hand side of the conditional independence in \eqref{eq:r_indep_given_mb}. Therefore,
by decomposition of
\eqref{eq:r_indep_given_mb},
\begin{equation}   
\label{eq:proba_r_given_mb}
P\left(
\mathbf R^{(a)}
\mid
\mathbb V^o,
\mathbb V^m,
\mathbf R^{(1)},\ldots,\mathbf R^{(a-1)}
\right)
=
P\left(
\mathbf R^{(a)}
\mid
\mathbf M_a
\right).
\end{equation}

Substituting \eqref{eq:proba_r_given_mb} into \eqref{eq:r_factorization} gives

\begin{equation}
\label{eq:r_factorization_r_given_mb}    
P
\left(
\mathbb R
\mid
\mathbb V^o,\mathbb V^m
\right)
=
\prod_{a=1}^{K}
P\left(
\mathbf R^{(a)}
\mid
\mathbf M_a
\right).
\end{equation}

Evaluating \eqref{eq:r_factorization_r_given_mb} at $\mathbb R=0$ yields

\begin{equation}
\label{eq:r_factorization_r_given_mb_split_o_and_m}    
P
\left(
\mathbb R=0
\mid
\mathbb V^o,\mathbb V^m
\right)
=
\prod_{a=1}^{K}
P\left(
\mathbf R^{(a)}=0
\mid
\mathbf M_a^o,\mathbf M_a^m
\right).
\end{equation}

By assumption,

$$
\mathbb R^{(a)}
\cap
R_{
MB^m
\left(
\mathbb R^{(a)}
\right)
}
=
\varnothing.
$$

Consequently, the missingness indicators associated with the partially
observed clusters in
$MB^m(\mathbb R^{(a)})$ belong to
$\mathbb R\setminus\mathbf R^{(a)}$. By the conditional independence
in \eqref{eq:r_indep_given_mb},

$$
\indepc{\mathbf R^{(a)}}
{R_{\mathbf M_a^m}}
{\mathbf M_a^o,\mathbf M_a^m}.
$$

It follows that

\begin{equation}    
\label{eq:proba_r_given_mb_split_o_and_m_rmb0}
P\left(
\mathbf R^{(a)}=0
\mid
\mathbf M_a^o,\mathbf M_a^m
\right)
=
P\left(
\mathbf R^{(a)}=0
\mid
\mathbf M_a^o,
\mathbf M_a^m,
R_{\mathbf M_a^m}=0
\right).
\end{equation}

Substituting \eqref{eq:proba_r_given_mb_split_o_and_m_rmb0} into \eqref{eq:r_factorization_r_given_mb_split_o_and_m} gives

\begin{equation}
\label{eq:r_factorization_r_given_mb_split_o_and_m_rmb0}    
P
\left(
\mathbb R=0
\mid
\mathbb V^o,\mathbb V^m
\right)
=
\prod_{a=1}^{K}
P\left(
\mathbf R^{(a)}=0
\mid
\mathbf M_a^o,
\mathbf M_a^m,
R_{\mathbf M_a^m}=0
\right).
\end{equation}

By the definition of conditional probability,

$$
P
\left(
\mathbb R=0,
\mathbb V^o,
\mathbb V^m
\right)
=
P
\left(
\mathbb V^o,\mathbb V^m
\right)
P
\left(
\mathbb R=0
\mid
\mathbb V^o,\mathbb V^m
\right).
$$

Using \eqref{eq:r_factorization_r_given_mb_split_o_and_m_rmb0} and the assumed positivity condition, we obtain

\begin{equation}    
\label{eq:recoverability_formula_variable_version}
P
\left(
\mathbb V^o,\mathbb V^m
\right)
=
\frac{
P
\left(
\mathbb R=0,
\mathbb V^o,
\mathbb V^m
\right)
}{
\displaystyle
\prod_{a=1}^{K}
P\left(
\mathbf R^{(a)}=0
\mid
\mathbf M_a^o,
\mathbf M_a^m,
R_{\mathbf M_a^m}=0
\right)
}.
\end{equation}

The numerator in \eqref{eq:recoverability_formula_variable_version} is observable from complete cases. In every
denominator factor, the event

$$
R_{\mathbf M_a^m}=0
$$

ensures that all partially observed variables in $\mathbf M_a^m$ equal
their observed proxies. Hence, each denominator factor is a functional
of the observed-data distribution.

The blocks and their substantive Markov blankets are defined entirely
from $\mathcal G^{\mathbb C,cm}$. Moreover, the conditional
independences used above hold in every compatible variable-level
m-ADMG by Lemma~\ref{lem:cluster_markov_blanket_transfer}. Therefore,
the resulting observed-data functional is the same for every
variable-level m-ADMG compatible with
$\mathcal G^{\mathbb C,cm}$.

Rewriting \eqref{eq:recoverability_formula_variable_version} in cluster notation gives

$$
\Pr(\mathbb c)
=
\frac{
\Pr
\left(
\mathbb R_{\mathbb C}=0,
\mathbb c
\right)
}{
\displaystyle
\prod_{a=1}^{K}
\Pr\left(
\mathbb R^{(a)}=0
\mid
MB^o
\left(
\mathbb R^{(a)}
\right),
MB^m
\left(
\mathbb R^{(a)}
\right),
R_{
MB^m
\left(
\mathbb R^{(a)}
\right)
}
=0
\right)
}.
$$

Thus, $\Pr(\mathbb C)$ is recoverable uniformly over all variable-level
m-ADMGs compatible with $\mathcal G^{\mathbb C,cm}$.

\paragraph{Extension to m-C-DMGs}

The proof for m-C-DMGs follows the same argument. It  follows by
projecting forbidden variable-level paths to the m-C-DMG and applying the
soundness of cluster d-separation to its cluster Markov blankets.
Therefore, the stated graphical conditions are  sufficient
for recovering the joint distribution in both cm-C-DMGs and m-C-DMGs.
\end{proof}

\subsection{Proof of Proposition~\ref{proposition:compatibility_recover_joint_distribution}}
\begin{proof}
Let \(\mathcal{G}^{m}_{\mathbb{C}}\) be an m-C-DMG and let
\(\mathcal{G}^{\mathbb{C},cm}=\pi(\mathcal{G}^{m}_{\mathbb{C}})\)
be the associated cm-C-DMG.

\textbf{Theorem~\ref{theorem:necessary_recovering_joint_dist} is satisfied in $\mathcal{G}^{\mathbb{C}, cm}  
\Rightarrow$
Theorem~\ref{theorem:necessary_recovering_joint_dist} is satisfied in $\mathcal{G}^{\mathbb{C}, m}$}

This implication is given by Proposition~\ref{proposition:compatibility_recover}.

\textbf{Theorem~\ref{theorem:necessary_recovering_joint_dist} is satisfied in $\mathcal{G}^{\mathbb{C}, cm}  
\Leftarrow$
Theorem~\ref{theorem:necessary_recovering_joint_dist} is satisfied in $\mathcal{G}^{\mathbb{C}, m}$}

Equivalently, we prove the contrapositive. Suppose that
the conditions of Theorem~\ref{theorem:necessary_recovering_joint_dist} are not satisfied in
\(\mathcal{G}^{\mathbb{C},cm}\). Then, there exists a cluster
\(C_{\mathbb{X}}\in\mathbb{C}^m\) such that \(C_{\mathbb{X}}\) and its
cluster-level missingness vertex \(R_{C_{\mathbb{X}}}\) are neighbors or
are connected by a path whose intermediate vertices are colliders and belong to
\(\mathbb{C}^m\cup\mathbb{C}^o\).

Suppose \(C_{\mathbb{X}}\) is adjacent to
\(R_{C_{\mathbb{X}}}\) in the cm-C-DMG. Since
\(R_{C_{\mathbb{X}}}\) is obtained by merging the missingness vertices
associated with variables in \(C_{\mathbb{X}}\), this adjacency must come
from an adjacency in the m-C-DMG between \(C_{\mathbb{X}}\) and at least one
of its variable-level missingness vertices, say \(R_i\).
But then, in \(\mathcal{G}^{m}_{\mathbb{C}}\), the cluster
\(C_{\mathbb{X}}\) and its missingness vertex \(R_i\) are neighbors. Thus the first condition
of Theorem~\ref{theorem:necessary_recovering_joint_dist} is violated in $\mathcal{G}^{m}_{\mathbb{C}}$. 

Suppose that \(C_{\mathbb{X}}\) is connected to
\(R_{C_{\mathbb{X}}}\) via a path in which all intermediate vertices are colliders and elements of $\mathbb{C}^m\cup \mathbb{C}^o$ in the cm-C-DMG. By the definition of a cm-C-DMG, all compatible m-C-DMGs must contain at least one variable \(R_{X_i}\in{R_{C_\mathbb{X}}}\) connected to \(C_{\mathbb{X}}\) via such a path. Violation of the second condition of Theorem~\ref{theorem:necessary_recovering_joint_dist} in a cm-C-DMG thus implies violation in all compatible m-C-DMGs. 

Therefore,
\[
\text{ Theorem~\ref{theorem:necessary_recovering_joint_dist} not satisfied in } \mathcal{G}^{\mathbb{C},cm}
\Rightarrow
\text{ Theorem~\ref{theorem:necessary_recovering_joint_dist} not satisfied in } \mathcal{G}^{m}_{\mathbb{C}}.
\]
Taking the contrapositive gives
\[
\text{ Theorem~\ref{theorem:necessary_recovering_joint_dist} is satisfied in } \mathcal{G}^{m}_{\mathbb{C}}
\Rightarrow
\text{ Theorem~\ref{theorem:necessary_recovering_joint_dist} is satisfied in } \mathcal{G}^{\mathbb{C},cm}.
\]

\textbf{Theorem~\ref{theorem:sufficient_recovering_joint_dist} is satisfied in $\mathcal{G}^{\mathbb{C}, cm}  
\Rightarrow$
Theorem~\ref{theorem:sufficient_recovering_joint_dist} is satisfied in $\mathcal{G}^{\mathbb{C}, m}$}

This implication is given by Proposition~\ref{proposition:compatibility_recover}.

\textbf{Theorem~\ref{theorem:sufficient_recovering_joint_dist} is satisfied in $\mathcal{G}^{\mathbb{C}, cm}  
\not\Leftarrow$
Theorem~\ref{theorem:sufficient_recovering_joint_dist} is satisfied in $\mathcal{G}^{\mathbb{C}, m}$}

The converse of the conditions of Theorem~\ref{theorem:sufficient_recovering_joint_dist} fails because merging missingness indicators can merge
distinct missingness blocks and thereby create a new violation of the
block-intersection condition.

For example, consider an m-C-DMG with missingness-indicator blocks
$$
\mathbb R^{(1)}=\{R_A,R_X\}
\qquad\text{and}\qquad
\mathbb R^{(2)}=\{R_B,R_Y\},
$$
where $X$ and $Y$ belong to the same substantive cluster $C$. Suppose
that the graph contains the collider paths
$$
R_A
\longdashleftrightarrow
C_{Z_1}
\longdashleftrightarrow
R_X
$$
and
$$
R_B
\longdashleftrightarrow
C_{Z_2}
\longdashleftrightarrow
R_Y,
$$
together with the edge
$$
C_B\longrightarrow R_A,
$$
and no other relevant edges.

In the m-C-DMG, $C_B$ belongs to the Markov blanket of
$\mathbb R^{(1)}$, but its indicator $R_B$ belongs to the distinct block
$\mathbb R^{(2)}$. Thus, the block-intersection condition can hold for
both blocks.

After projection, $R_X$ and $R_Y$ are merged into the single vertex
$R_C$. The projected graph contains the collider path
$$
R_A
\longdashleftrightarrow
C_{Z_1}
\longdashleftrightarrow
R_C
\longdashleftrightarrow
C_{Z_2}
\longdashleftrightarrow
R_B.
$$
Hence, $R_A$, $R_C$, and $R_B$ belong to the same
missingness-indicator block. Since $C_B\to R_A$, the cluster $C_B$
belongs to the Markov blanket of this block, while $R_B$ belongs to the
block itself. The block-intersection condition therefore fails in the
cm-C-DMG.

Thus, the sufficient conditions may hold in an m-C-DMG but fail in its
associated cm-C-DMG.
\end{proof}

\subsection{Proof of Theorem~\ref{theorem:soundness_do_calculus_missingness}}

\begin{lemma}[Soundness of do-calculus in cluster missingness graphs]
\label{lemma:do-calculus_missingness}
Let \(\mathcal{G}^{\mathbb{C},*m}\) denote either an m-C-DMG or a cm-C-DMG, and let
\(\compatible{\mathcal{G}^{\mathbb{C},*m}}\) be the set of variable-level
m-ADMGs compatible with it. Then any application of a rule of do-calculus in
\(\mathcal{G}^{\mathbb{C},*m}\) is valid in every
\(\mathcal{G}\in\compatible{\mathcal{G}^{\mathbb{C},*m}}\).
\end{lemma}

\begin{proof}
Both m-C-DMGs and cm-C-DMGs are C-DMGs whose vertices represent either clusters of
substantive variables, missingness indicators, missingness-indicator clusters,
or proxy variables. Hence the soundness of do-calculus in C-DMGs established by
\citet{Ferreira_2025} applies directly.

More precisely, if a rule of do-calculus is applicable in
\(\mathcal{G}^{\mathbb{C},*m}\), then the corresponding d-separation condition
holds in the relevant mutilated cluster graph. By the soundness result of
\citet{Ferreira_2025}, this d-separation condition entails the corresponding
d-separation condition in every compatible variable-level m-ADMG. The standard
soundness of do-calculus at the variable level then implies that the associated
do-calculus equality is valid in every
\(\mathcal{G}\in\compatible{\mathcal{G}^{\mathbb{C},*m}}\).

Therefore, any do-calculus transformation licensed by
\(\mathcal{G}^{\mathbb{C},*m}\) is sound for all compatible variable-level
missingness graphs.
\end{proof}

Now we proceed to the proof of the theorem.
\begin{proof}
Let \(\compatible{\mathcal{G}^{\mathbb{C},*m}}\) denote the set of variable-level
m-ADMGs compatible with \(\mathcal{G}^{\mathbb{C},*m}\). Suppose that the query
\(Q\) can be transformed into an expression \(F\) using do-calculus, probability
manipulations, and Equation~\ref{eq:from_missing_to_not_missing}, where \(F\)
contains no do-operators and involves only observable quantities.

By Lemma~\ref{lemma:do-calculus_missingness}, each do-calculus step
performed in \(\mathcal{G}^{\mathbb{C},*m}\) is sound in every variable-level
m-ADMG
\[
\mathcal{G}\in \compatible{\mathcal{G}^{\mathbb{C},*m}}.
\]
Moreover, the probability manipulations are algebraic identities and are
therefore valid in every compatible model. Equation~\ref{eq:from_missing_to_not_missing}
is the consistency relation defining the proxy variables, and is also valid in
every compatible missingness model.

Thus, the whole sequence of transformations from \(Q\) to \(F\) is valid in
every variable-level m-ADMG compatible with
\(\mathcal{G}^{\mathbb{C},*m}\). Hence, for every such compatible model,
\[
Q = F.
\]

Since \(F\) contains no do-operators and involves only observable quantities, it
is a functional of the observed-data distribution. Therefore, \(Q\) is uniquely
determined from the observed-data distribution for all variable-level m-ADMGs
compatible with \(\mathcal{G}^{\mathbb{C},*m}\).

By definition, \(Q\) is recoverable from
\(\mathcal{G}^{\mathbb{C},*m}\).
\end{proof}

\section{Relationship between the necessary and sufficient graphical conditions}

The block-intersection condition appearing in the sufficient
recoverability result in Tehorem~\ref{theorem:sufficient_recovering_joint_dist} implies the  graphical conditions appearing
in the necessary recoverability result in Theorem~\ref{theorem:necessary_recovering_joint_dist}. This implication follows directly
from the definition of the substantive Markov blanket of a
missingness-indicator block. It holds both at the variable level, for
m-ADMGs, and at the cluster level, for m-C-DMGs and cm-C-DMGs.

The converse does not hold in general. In particular, the
block-intersection condition may fail because a partially observed
substantive vertex is connected to its own missingness-indicator block
through a collider path containing another missingness-indicator vertex.
Such a path is not excluded by the necessary condition, which only considers
paths whose intermediate vertices are substantive colliders.

\begin{proposition}[The sufficient block condition implies the necessary conditions]
\label{prop:block_condition_implies_necessary}
Let $\mathcal G$ be either:
\begin{enumerate}
\item a variable-level m-ADMG, in which case let
$\mathbb A^m=\mathbb V^m$ and associate with each
$A=X\in\mathbb A^m$ its missingness indicator
$R_A=R_X$; or
\item an m-C-DMG or a cm-C-DMG, in which case let
$\mathbb A^m=\mathbb C^m$. For an m-C-DMG, associate with each
partially observed variable $X$ the pair $(C_X,R_X)$, and for a
cm-C-DMG associate with each $A=C\in\mathbb C^m$ the
cluster-level missingness indicator $R_A=R_C$.
\end{enumerate}

Let
$
\mathbb R^{(1)},\ldots,\mathbb R^{(K)}
$
be the missingness-indicator blocks of $\mathcal G$, and let
$MB^m(\mathbb R^{(i)})$ denote the partially observed substantive
part of the Markov blanket of $\mathbb R^{(i)}$.
Suppose that
\begin{equation}
\label{eq:additional_condition_for_sufficient_condition_recoverability_joint_dist}
\mathbb R^{(i)}
\cap
R_{MB^m(\mathbb R^{(i)})}
=
\varnothing
\qquad
\text{for every }i.
\end{equation}

Then, for every partially observed substantive vertex $A$, the following
conditions hold:

\begin{enumerate}
\item $A$ and its missingness indicator $R_A$ are not adjacent;

\item there is no path between $A$ and $R_A$ whose intermediate
vertices are colliders and substantive vertices.
\end{enumerate}
\end{proposition}

\begin{proof}
We give the proof simultaneously for variable-level and cluster-level
graphs.

Fix a partially observed substantive vertex $A$, and let
$\mathbb R(A)$ denote the missingness-indicator block containing $R_A$.

Suppose first that $A$ and $R_A$ are adjacent. Since
$R_A\in\mathbb R(A)$, the definition of the substantive Markov blanket
of a missingness-indicator block implies that
$
A\in MB^m\left(\mathbb R(A)\right).
$
Consequently,
$
R_A\in
R_{MB^m(\mathbb R(A))}.
$
Since $R_A\in\mathbb R(A)$ as well, it follows that
$
R_A\in
\mathbb R(A)
\cap
R_{MB^m(\mathbb R(A))},
$
contradicting~\eqref{eq:additional_condition_for_sufficient_condition_recoverability_joint_dist}. Therefore, $A$ and $R_A$ cannot be
adjacent.

Suppose next that there exists a path
$
\pi
=
\langle
A,U_1,\ldots,U_k,R_A
\rangle
$
such that every intermediate vertex $U_j$ is a substantive collider.
The path $\pi$ is therefore a collider path between the substantive
vertex $A$ and the member $R_A$ of the block $\mathbb R(A)$. By the
definition of the substantive Markov blanket,
$
A\in MB^m\left(\mathbb R(A)\right).
$
As above, this implies
$
R_A\in
\mathbb R(A)
\cap
R_{MB^m(\mathbb R(A))},
$
again contradicting~\eqref{eq:additional_condition_for_sufficient_condition_recoverability_joint_dist}.

Hence neither adjacency nor a substantive-intermediate collider path can
exist between a partially observed substantive vertex and its own
missingness indicator.
\end{proof}

\begin{remark}
The proposition does not require the absence of edges between distinct
missingness-indicator vertices. That assumption may be required for the
factorization used in a sufficient recoverability theorem, but it is not
needed to prove that the block-intersection condition implies the
individual necessary graphical conditions.
\end{remark}

It is important to note that the  converse of Proposition~\ref{prop:block_condition_implies_necessary} does not hold. We will show this with a counterexample.

The absence, for every partially observed substantive vertex $A$, of

\begin{example}
First consider a variable-level m-ADMG containing partially observed
variables $X$ and $Y$, a substantive variable $Z$, and the edges
$$
X\longrightarrow R_Y,
\qquad
R_Y\longdashleftrightarrow Z,
\qquad
Z\longdashleftrightarrow R_X,
$$
with no other relevant edges.

There is no edge between the distinct missingness indicators $R_X$ and
$R_Y$. Moreover, $X$ and $R_X$ are not adjacent. The path
$$
X
\longrightarrow
R_Y
\longdashleftrightarrow
Z
\longdashleftrightarrow
R_X
$$
does not violate the individual necessary condition because one of its
intermediate vertices, namely $R_Y$, is a missingness indicator rather
than a substantive vertex. Assume that $Y$ is otherwise disconnected
from $R_Y$, so the individual necessary conditions also hold for $Y$.

On the other hand,
$$
R_Y
\longdashleftrightarrow
Z
\longdashleftrightarrow
R_X
$$
is a collider path between $R_Y$ and $R_X$. Hence they belong to the
same missingness-indicator block, say
$$
\mathbb R^{(1)}={R_X,R_Y}.
$$
Since $X\to R_Y$, the substantive vertex $X$ belongs to the partially
observed Markov blanket of this block:
$$
X\in MB^m(\mathbb R^{(1)}).
$$
Therefore,
$$
R_X\in R_{MB^m(\mathbb R^{(1)})}.
$$
Since $R_X\in\mathbb R^{(1)}$, we obtain
$$
R_X\in
\mathbb R^{(1)}
\cap
R_{MB^m(\mathbb R^{(1)})},
$$
and the block-intersection condition fails.

The cluster-level counterexample is obtained by replacing the substantive
vertices by their corresponding cluster vertices:
$$
C_{\mathbb X}\longrightarrow R_{C_{\mathbb Y}},
\qquad
R_{C_{\mathbb Y}}
\longdashleftrightarrow C_{\mathbb Z},
\qquad
C_{\mathbb Z}
\longdashleftrightarrow R_{C_{\mathbb X}}.
$$
Let $C_{\mathbb X}$ and $C_{\mathbb Y}$ be partially observed clusters,
and assume that $C_{\mathbb Y}$ is otherwise disconnected from
$R_{C_{\mathbb Y}}$.

There is no edge between the two missingness-indicator vertices. The path
from $C_{\mathbb X}$ to $R_{C_{\mathbb X}}$ contains the intermediate
missingness-indicator vertex $R_{C_{\mathbb Y}}$, and therefore does not
violate the necessary condition requiring all intermediate vertices to
belong to $\mathbb C^o\cup\mathbb C^m$. Nevertheless,
$R_{C_{\mathbb X}}$ and $R_{C_{\mathbb Y}}$ belong to the same
missingness-indicator block, while
$$
C_{\mathbb X}
\in
MB^m\left(
{R_{C_{\mathbb X}},R_{C_{\mathbb Y}}}
\right).
$$
Thus,
$$
R_{C_{\mathbb X}}
\in
{R_{C_{\mathbb X}},R_{C_{\mathbb Y}}}
\cap
R_{
MB^m(
{R_{C_{\mathbb X}},R_{C_{\mathbb Y}}}
)
},
$$
and the cluster-level block-intersection condition also fails.
\end{example}

The counterexample shows that the block-intersection condition is strictly
stronger than the individual necessary graphical conditions. The
difference arises because the block condition detects collider
connections containing intermediate missingness-indicator vertices,
whereas the individual necessary condition only excludes paths whose
intermediate vertices are substantive colliders.

\section{Simulated data}

The m-ADMGs in Figure~\ref{fig:data_admgs}, both of which are compatible with the cm-C-DMG in Figure \ref{fig:m_clusters}b, were used as the true data-generating mechanisms behind the simulated datasets.

\begin{figure}
    \centering
\begin{subfigure}{0.45\textwidth}
\begin{tikzpicture}[
    every node/.style={draw, circle, align=center, minimum width=0.2cm, minimum height=0.2cm},
    font=\footnotesize, scale=0.8
]

\node (Z1) at (-2,1.5) {$Z_1$};
\node (Z2) at (-2,0) {$Z_2$};
\node (X1) at (0, 1.5) {$X_1$};
\node (X2) at (0, 0) {$X_2$};

\node (Y) at (2.5, 0) {$Y$};
\node (Y*) at (-1, -3) {$Y^*$};
\node (ry) at (-4,0) {$R_Y$};

\node (X1*) at (1.25, 3) {$X_1^*$};
\node (rx1) at (-0.75,3) {$R_{X_1}$};
\node (rx2) at (-0.75, -1.25) {$R_{X_2}$};
\node (X2*) at (1.25, -1.25) {$X_2^*$};

\draw[->,>=latex] (Z1) -- (Z2);
\draw[->,>=latex] (Z1) -- (X1);
\draw[->,>=latex] (X1) -- (Z2);
\draw[->,>=latex] (X1) -- (X2);
\draw[->,>=latex] (Z1) -- (rx1);
\draw[->,>=latex] (rx1) -- (X1*);
\draw[->,>=latex] (X1) -- (X1*);
\draw[->,>=latex] (rx2) -- (X2*);
\draw[->,>=latex] (X2) -- (X2*);
\draw[->,>=latex] (ry) -- (Y*);
\draw[->,>=latex] (Y) to[out=-90,in=0] (Y*);
\draw[->,>=latex] (X1) -- (Y);

\draw[<->,>=latex, dashed] (ry) to [out=80,in=100, looseness=1] (Z1);

\end{tikzpicture}
\caption{Setting 1}
\end{subfigure}
\hfill 
\begin{subfigure}{0.45\textwidth}
\begin{tikzpicture}[
    every node/.style={draw, circle, align=center, minimum width=0.2cm, minimum height=0.2cm},
    font=\footnotesize, scale=0.8
]

\node (Z1) at (-2,1.5) {$Z_1$};
\node (Z2) at (-2,0) {$Z_2$};
\node (X1) at (0, 1.5) {$X_1$};
\node (X2) at (0, 0) {$X_2$};

\node (Y) at (2.5, 0.75) {$Y$};
\node (Y*) at (-0.75, -3) {$Y^*$};
\node (ry) at (-4,0) {$R_Y$};

\node (X1*) at (1.25, 3) {$X_1^*$};
\node (rx1) at (-0.75,3) {$R_{X_1}$};
\node (rx2) at (-0.75, -1.25) {$R_{X_2}$};
\node (X2*) at (1.25, -1.25) {$X_2^*$};

\draw[->,>=latex] (Z2) -- (Z1);
\draw[->,>=latex] (Z2) -- (X2);
\draw[->,>=latex] (Z2) -- (X1);
\draw[->,>=latex] (X2) -- (X1);
\draw[->,>=latex] (X2) -- (Z1);
\draw[->,>=latex] (X1) -- (Z1);

\draw[->,>=latex] (Z1) -- (rx1);
\draw[->,>=latex] (Z2) -- (rx1);
\draw[->,>=latex] (rx1) -- (X1*);
\draw[->,>=latex] (X1) -- (X1*);

\draw[->,>=latex] (Z1) -- (rx2);
\draw[->,>=latex] (Z2) -- (rx2);
\draw[->,>=latex] (rx2) -- (X2*);
\draw[->,>=latex] (X2) -- (X2*);

\draw[->,>=latex] (ry) -- (Y*);
\draw[->,>=latex] (Y) to[out=-90,in=0] (Y*);

\draw[->,>=latex] (X1) -- (Y);
\draw[->,>=latex] (X2) -- (Y);

\draw[<->,>=latex, dashed] (ry) to [out=80,in=100, looseness=1] (Z1);
\draw[<->,>=latex, dashed] (ry) to [out=80,in=100, looseness=1] (Z2);

\end{tikzpicture}
\caption{Setting 2}
\end{subfigure}
    \caption{m-ADMGs compatible with Figure \ref{fig:m_clusters}b used to generate simulated data.}
    \label{fig:data_admgs}
\end{figure}

Missingness was distributed approximately as follows across simulated datasets. In Setting 1, $X_1$ was around $35\%$ missing, $X_2$ was around $20\%$ missing, and $Y$ was around 55\% missing. In Setting 2, $X_1$ was around $25\%$ missing, $X_2$ was around $65\%$ missing, and $Y$ was around $55\%$ missing.

Data was simulated and analyzed using a MacBook Pro, with negligible execution time (under 1 minute for the whole script).



\end{document}